\documentclass[journal,twocolumn,10pt,twoside]{IEEEtranTCOM}
\newtheorem{proposition}{Proposition}
\newtheorem{theorem}{Theorem}
\newtheorem{remark}{Remark}

\usepackage{cite}
\usepackage{amsmath,amssymb,amsfonts}
\usepackage{algorithmic}
\usepackage{graphicx}
\usepackage{textcomp}
\usepackage{xcolor}
\usepackage{subcaption}  
\usepackage{booktabs}   
\usepackage{hyperref}
\normalsize

\ifCLASSINFOpdf
\else
\fi

\begin{document}
		
		\title{Fundamental Performance Limits of Non-Coherent ISAC: A Data-Aided Sensing Perspective}
		
		\author{Dongsheng Peng,~\IEEEmembership{Member,~IEEE,} Chengkai Zhao, Yihong Li,~\IEEEmembership{Member,~IEEE,} Zhiqing Wei, \\
			Jun Chen,~\IEEEmembership{Senior~Member,~IEEE} and Ping Chen~\IEEEmembership{Senior~Member,~IEEE}
			
		}
		
		\maketitle

		\begin{abstract}
			In this paper, we investigate a bistatic multiple-input multiple-output (MIMO) integrated sensing and communication (ISAC) system over block-fading channels, focusing on the scenario where the sensing and communication receivers (Rxs) are co-located. Under the assumption of unknown channel state information (CSI) at the Rx, two schemes are considered: pilot sensing (PS) and data-aided sensing (DAS). The communication rate-sensing distortion functions for both schemes are characterized. For the DAS scheme, a closed-form asymptotic expression for the sensing distortion is derived by using random matrix theory (RMT).
			The asymptotic performance analysis explicitly quantifies the significant gains of the DAS scheme, revealing a strict $3$ dB effective SNR improvement in the low-SNR regime and a strictly faster performance scaling rate in the high-SNR limit compared to the PS scheme.
		\end{abstract}


		%
		\IEEEpeerreviewmaketitle

		\section{Introduction}
		
		Integrated sensing and communications (ISAC) \cite{liu2022isac} has emerged as a key enabling technology for the sixth-generation (6G) wireless networks, empowering a single system to perform both sensing and communication tasks over a shared spectrum and hardware platform. By exploiting the synergy between the two functionalities, ISAC yields substantial gains in both spectral and hardware efficiency \cite{liu2022survey_comst}. These inherent advantages position ISAC as a foundational technology for emerging scenarios such as vehicle-to-everything (V2X), the low-altitude economy, and smart cities \cite{chafii2023twelve, saad2019vision}, where real-time environmental awareness must coexist with high-throughput communication services.
		
		Current research on ISAC primarily revolves around three design paradigms: communication-centric \cite{sturm2011waveform,wu2022integrating,wu2024joint,hua2024integrated}, radar-centric \cite{hassanien2016dual,huang2020majorcom,temiz2023experimental}, and joint waveform design \cite{liu2022transmit,liu2020joint,ren2024fundamental,hua2024mimo}. While the latter two paradigms offer favorable sensing performance or flexible performance tradeoffs, they necessitate modifications to existing waveform structures or introduce additional signal design complexity. Consequently, they face certain challenges regarding deployment costs and standard compatibility. In contrast, communication-centric ISAC designs directly repurpose standard communication waveforms to execute sensing tasks \cite{liu2025sensing}. Driven by their core advantages of low implementation overhead and full compatibility with existing cellular infrastructure, this approach has become a prominent research hotspot in the academic community in recent years.
		
		Under the communication-centric ISAC framework, existing studies primarily rely on reference signals (i.e., pilots) embedded in standard communication frame structures to perform sensing tasks \cite{li2026rethinking}. Since pilot signals occupy only a limited portion of time-frequency resources within the overall frame structure, typically ranging from 0.15\% to 25\%, this constrains the sensing resolution and detection accuracy. A natural enhancement is to incorporate the data payload signals, which account for the vast majority of system resources, into the sensing processing pipeline, namely, data-aided sensing (DAS). However, data payload signals are inherently random and are not designed for sensing purposes, which introduces new challenges to the design and performance analysis of ISAC systems.
		
		Currently, the majority of research on DAS focuses on spatial-domain design \cite{he2024mse,xu2025exploiting,xie2026sensing}, which aims to enhance ISAC performance by optimizing precoding matrices or beamforming strategies under predetermined frame structures. In addition, there are several task-oriented DAS studies. For example, for monostatic sensing in frequency-division duplex (FDD) systems, \cite{hua2024integrated} proposed the simultaneous utilization of both pilots and data payloads for target detection. In bistatic ISAC scenarios, the problem of target detection jointly utilizing pilots and data payloads was investigated in \cite{xie2025bistatic}, while \cite{yu2026sensing} focused on repurposing uplink data for multi-source localization in pilot-free scenarios.
		Compared with spatial precoding optimization, directly quantifying the impact of pilot training duration and power allocation between pilots and data provides a more fundamental information-theoretic understanding of the tradeoff between communication and sensing. It is worth noting that, in traditional multi-antenna communication systems, capacity optimization problems involving training duration and pilot/data power allocation have been well studied \cite{hassibi2003training}. However, in ISAC scenarios employing DAS, since the data payload simultaneously serves the dual functions of communication and sensing, a systematic analysis of the joint impact of pilot training duration and pilot/data power allocation on the communication and sensing tradeoff performance is still lacking.
		
		Against this backdrop, this paper investigates the joint optimization of pilot training duration and power allocation under DAS, with the goal of revealing their impact on the tradeoff between communication and sensing.
		Specifically, we consider a bistatic ISAC scenario where the sensing receiver (Rx) is co-located with the communication Rx, and the sensing task focuses on the estimation of the target response matrix (TRM). Our main contributions are summarized as follows:
		
		\begin{itemize}
			\item To address the challenge that the sensing distortion is analytically intractable due to the random data matrix in the DAS scheme, we employ random matrix theory (RMT) to derive a closed-form asymptotic expression for the mean squared error (MSE) of the TRM estimation.
			
			\item We derive the closed-form expressions of the communication rate-sensing distortion functions $R(D)$ for both the pilot-only sensing (PS) and DAS schemes. This reveals the analytical relationship among the optimal pilot duration, the power allocation strategy, and the achievable communication–sensing tradeoff for each scheme.

			\item We derive closed-form asymptotic expressions for the effective signal-to-noise ratio (SNR) of both schemes in the low-SNR and high-SNR regimes. The analytical results explicitly quantify the fundamental superiority of the DAS scheme over the PS scheme, demonstrating a strict $3$~dB effective SNR gain in the low-SNR regime and a fundamentally faster asymptotic convergence rate ($\mathcal{O}(1/K^2)$ versus $\mathcal{O}(1/\sqrt{K})$) in the high-SNR limit. Furthermore, under asymptotic conditions (i.e., in the limit of large coherence blocks or high SNRs), the communication rate and sensing distortion of the DAS scheme become decoupled, whereas the PS scheme seamlessly recovers the classical findings reported in \cite{hassibi2003training}\footnote{In the high-SNR limit, the channel estimation approaches perfection. Consequently, the sensing distortion constraint is automatically satisfied.}.
		\end{itemize}
		
		The remainder of this paper is organized as follows. Section \ref{sec:system_model} defines the system model and the performance evaluation metrics. Section \ref{sec:pilot_sensing} and Section \ref{sec:data_assisted_sensing} derive the $R(D)$ functions and the optimal resource allocation strategies for the PS and DAS schemes, respectively. Section \ref{sec:asymptotic_analysis} presents the theoretical analysis in asymptotic regimes. Section \ref{sec:numerical_results} provides the numerical simulation results. Finally, Section \ref{sec:conclusion} concludes the paper.

		\textit{Notation:} $\mathbf{X}$ and $\mathbf{x}$ denote a matrix and a column vector, respectively, while $x$ or $X$ denotes a scalar. $\mathbb{C}^{M \times N}$ represents the space of $M \times N$ complex-valued matrices, and $\mathbb{R}^+$ denotes the set of positive real numbers. For a matrix $\mathbf{A}$, $\mathbf{A}^T$, $\mathbf{A}^*$, $\mathbf{A}^H$, and $\mathbf{A}^{-1}$ denote its transpose, conjugate, conjugate transpose, and inverse, respectively. $\mathbf{A} \succeq 0$ indicates that $\mathbf{A}$ is a positive semi-definite matrix. $\operatorname{Tr}(\mathbf{A})$ and $\det(\mathbf{A})$ stand for the trace and determinant of the matrix; $\operatorname{vec}(\mathbf{A})$ is the column-wise vectorization operator; $\operatorname{diag}(a_1,\ldots,a_M)$ represents a diagonal matrix with diagonal elements $a_1,\ldots,a_M$; $\|\mathbf{A}\|_F$ denotes the Frobenius norm. $\otimes$ represents the Kronecker product. $\mathbf{I}_N$ is the $N \times N$ identity matrix. $\mathbb{E}[\cdot]$ is the statistical expectation operator. $\mathcal{CN}(\boldsymbol{\mu}, \mathbf{\Sigma})$ denotes the circularly symmetric complex Gaussian distribution with mean $\boldsymbol{\mu}$ and covariance matrix $\mathbf{\Sigma}$. $\mathcal{O}(\cdot)$ stands for the standard big-O notation for asymptotic expansions.

		\section{System Model}\label{sec:system_model}
		
		We consider a point-to-point (P2P) bistatic ISAC system, as illustrated in Fig. \ref{fig:system_model}. The system consists of a multi-antenna ISAC transmitter (Tx) and a multi-antenna ISAC Rx. In the considered configuration, the Rx adopts a co-located architecture, comprising a communication receiver (Rx1) and a sensing receiver (Rx2), which are responsible for the communication data demodulation and target sensing tasks, respectively.
		
		Specifically, the system operates in a non-coherent block fading mode, implying that the multiple-input multiple-output (MIMO) channel matrix $\mathbf{H}$ is completely unknown to both the Tx and the Rx at the beginning of each coherence block. Under this setup, $\mathbf{H}$ carries a dual physical significance: for Rx1, $\mathbf{H}$ represents the fading channel for data transmission, whereas for Rx2, $\mathbf{H}$ serves as the TRM that characterizes the environmental scattering features. 
		There exists a distinct difference in their requirements for estimating $\mathbf{H}$: Rx1 relies on a basic channel estimation to accomplish data demodulation, whereas Rx2 must achieve a significantly higher-precision channel reconstruction to satisfy the target sensing demands.
		\begin{figure}[t]
			\centerline{\includegraphics[width=8.5cm]{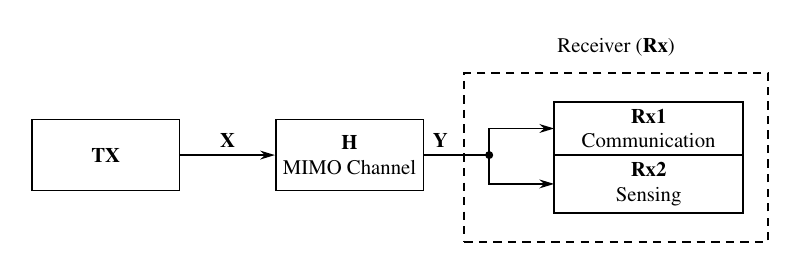}} \caption{P2P MIMO ISAC system.} 
			\label{fig:system_model}
		\end{figure}

		\subsection{Signal Model}\label{subsec:signal_model}
		
		Consider the aforementioned MIMO ISAC system equipped with $M$ transmit antennas and $N$ receive antennas. It is assumed that the channel follows a quasi-static block-fading model \cite{xiong2023fundamental}, meaning that the channel matrix remains constant within a coherence block of length $T$ symbols, and varies independently and randomly across different blocks. Within a single coherence block, the received signal model is given by
		\begin{align}
			\mathbf{Y} &= \mathbf{H}\mathbf{X} + \mathbf{Z}, \label{eq:mimo_model}
		\end{align}
		where $\mathbf{Y} \in \mathbb{C}^{N \times T}$ and $\mathbf{X} \in \mathbb{C}^{M \times T}$ represent the received and transmitted signals, respectively. The matrix $\mathbf{Z} \in \mathbb{C}^{N \times T}$ denotes the additive white Gaussian noise (AWGN), whose entries are independent and identically distributed (i.i.d.) following $\mathcal{CN}(0, \sigma_z^2)$. Furthermore, $\mathbf{H} \in \mathbb{C}^{N \times M}$ represents the communication channel matrix/TRM, with its entries being i.i.d. $\sim \mathcal{CN}(0, 1)$.

		Specifically, the transmitted signal matrix $\mathbf{X}$ is composed of two parts: pilots and data, which can be expressed as
		\begin{align}
			\mathbf{X} = [\mathbf{X}_{\tau}, \mathbf{X}_{d}],
		\end{align}
		where $\mathbf{X}_{\tau} \in \mathbb{C}^{M \times T_{\tau}}$ denotes the deterministic pilot signals transmitted during the channel estimation phase, intended for channel estimation/TRM estimation. The matrix $\mathbf{X}_{d} \in \mathbb{C}^{M \times T_{d}}$ represents the data signals transmitted during the data transmission phase for communication purposes. The durations of the two phases satisfy the time constraint
		\begin{align}
			T_{\tau} + T_{d} = T. \label{eq:time_partition}
		\end{align}
		\begin{remark}
			The condition $T_\tau \geq M$ is necessary to ensure that the channel matrix $\mathbf{H}$ can be effectively estimated, as the number of observations must not be less than the number of unknown parameters.
		\end{remark}
		Accordingly, the received signals for the two phases are given by
		\begin{align}
			\mathbf{Y}_\tau &= \mathbf{H}\mathbf{X}_\tau + \mathbf{Z}_\tau, \label{eq:pilot_model} \\
			\mathbf{Y}_d    &= \mathbf{H}\mathbf{X}_d    + \mathbf{Z}_d.    \label{eq:data_model}
		\end{align}
		The transmit power constraints for the two phases are defined as $\operatorname{Tr}(\mathbf{X}_\tau \mathbf{X}_\tau^H) = P_\tau T_\tau$ and $\mathbb{E}[\operatorname{Tr}(\mathbf{X}_d \mathbf{X}_d^H)] = P_d T_d$, respectively. Under a total power budget $P$, the overall power constraint of the system is given by
		\begin{align}
			P_\tau T_\tau + P_d T_d \leq PT. \label{eq:power_budget}
		\end{align}
		Let us define the SNR for each phase and the total system SNR as $\rho_\tau := P_\tau/\sigma_z^2$, $\rho_d := P_d/\sigma_z^2$, and $\rho := P/\sigma_z^2$, respectively. Consequently, \eqref{eq:power_budget} can be equivalently rewritten as
		\begin{align}
			\rho_\tau T_\tau + \rho_d T_d \leq \rho T. \label{eq:snr_constraint}
		\end{align}
		
		\subsection{Performance Metrics}\label{subsec:performance_metric}
		\subsubsection{Sensing Metric}
		
		In this paper, we adopt the MSE of the TRM estimation to evaluate the sensing performance, which is defined as
		\begin{align}
			\mathrm{MSE} := \mathbb{E}\!\left[\|\mathbf{H} - \hat{\mathbf{H}}\|_F^2\right],
			\label{eq:mse_def}
		\end{align}
		where $\hat{\mathbf{H}}$ is the estimate of $\mathbf{H}$. By utilizing the Kronecker product identity, \eqref{eq:pilot_model} can be equivalently vectorized as
		\begin{align}\label{eq:pilot_vec}
			\mathbf{y} = \tilde{\mathbf{X}} \mathbf{h} + \mathbf{z}_\tau,
		\end{align}
		where $\mathbf{y} := \mathrm{vec}(\mathbf{Y}_\tau)$, $\mathbf{h} := \mathrm{vec}(\mathbf{H})$, $\mathbf{z}_\tau := \mathrm{vec}(\mathbf{Z}_\tau)$, and $\tilde{\mathbf{X}} := \mathbf{X}_\tau^T \otimes \mathbf{I}_N$.  The covariance matrix of the estimation error $\tilde{\mathbf{h}} := \mathbf{h} - \hat{\mathbf{h}}$ for the linear minimum mean squared error (LMMSE) estimate $\hat{\mathbf{h}}$ of $\mathbf{h}$ based on $\mathbf{y}$
		is given by
		\begin{align}
			\mathbf{R}_{\tilde{\mathbf{h}}} = \left(\mathbf{R}_h^{-1} + \tilde{\mathbf{X}}^H \mathbf{R}_z^{-1} \tilde{\mathbf{X}}\right)^{-1},
			\label{eq:error_cov}
		\end{align}
		where $\mathbf{R}_h$ and $\mathbf{R}_z$ are the prior covariance matrices of the channel vector and the noise vector, respectively. Based on the statistical assumptions $\mathbf{R}_h = \mathbf{I}_{NM}$ and $\mathbf{R}_z = \sigma_z^2 \mathbf{I}_{NT_\tau}$, substituting these into \eqref{eq:error_cov} simplifies the sensing MSE as follows:
		\begin{align}
			\mathrm{MSE}
			&= \operatorname{Tr}(\mathbf{R}_{\tilde{\mathbf{h}}}) \nonumber \\
			&= N \cdot \operatorname{Tr}\!\left[\left(\mathbf{I}_M + \frac{1}{\sigma_z^2}\mathbf{X}_\tau \mathbf{X}_\tau^H\right)^{-1}\right].
			\label{eq:mse_expr}
		\end{align}
		Subject to the power constraint $\operatorname{Tr}(\mathbf{X}_\tau \mathbf{X}_\tau^H) \leq P_\tau T_\tau$, the optimal pilot matrix that minimizes \eqref{eq:mse_expr} and the corresponding minimum MSE are given by the following proposition\cite{hassibi2003training}.
		\begin{proposition}
			\label{prop:optimal_pilot}
			The optimal pilot matrix that minimizes the sensing MSE satisfies
			\begin{align}
				\mathbf{X}_\tau \mathbf{X}_\tau^H = \frac{P_\tau T_\tau}{M}\,\mathbf{I}_M,
				\label{eq:optimal_pilot}
			\end{align}
			which implies that the pilot matrix is an orthogonal matrix with equal power allocation. The corresponding minimum sensing MSE is given by
			\begin{align}
				\mathrm{MSE}^* = \frac{NM}{1 + \rho_\tau T_\tau / M}.
				\label{eq:mse_opt}
			\end{align}
		\end{proposition}
		\begin{IEEEproof}
			Please refer to Appendix \ref{app:A} .
		\end{IEEEproof}

		\subsubsection{Communication Metric}
		
		According to the properties of LMMSE estimation, the channel matrix can be decomposed as $\mathbf{H} = \hat{\mathbf{H}} + \tilde{\mathbf{H}}$, where the estimate $\hat{\mathbf{H}}$ and the estimation error $\tilde{\mathbf{H}}$ are independent of each other. Consequently, the received signal during the data phase can be rewritten as
		\begin{align}
			\mathbf{Y}_d &= \hat{\mathbf{H}}\mathbf{X}_d + \underbrace{\tilde{\mathbf{H}}\mathbf{X}_d + \mathbf{Z}_d}_{\mathbf{Z}_{\text{eff}}}, \label{eq:effective_model}
		\end{align}
		where the effective noise $\mathbf{Z}_{\text{eff}}$ comprises the AWGN $\mathbf{Z}_d$ and the interference term $\tilde{\mathbf{H}}\mathbf{X}_d$ resulting from the channel estimation error $\tilde{\mathbf{H}}$. Since $\mathbf{X}_d$, $\tilde{\mathbf{H}}$, and $\mathbf{Z}_d$ are mutually independent, it is straightforward to show that the covariance matrix of $\mathbf{Z}_{\text{eff}}$ is given by $\mathbf{R}_{Z_{\text{eff}}} = (P_d\sigma_{\tilde{H}}^2 + \sigma_z^2)\mathbf{I}_N$, where $\sigma_{\tilde{H}}^2$ denotes the estimation error variance of a single channel element.
		
		Assume that the data signals employ i.i.d. Gaussian signaling, meaning that $\mathbf{R}_{X_d} := \frac{1}{T_d}\mathbb{E}[\mathbf{X}_d \mathbf{X}_d^H] = \frac{P_d}{M}\mathbf{I}_M$. According to the worst-case uncorrelated additive noise theorem \cite{hassibi2003training}, by treating $\mathbf{Z}_{\mathrm{eff}}$ as an equivalent independent Gaussian noise with the same covariance matrix, a lower bound on the ergodic achievable rate can be obtained as
		\begin{align}
			R &= \frac{T - T_\tau}{T} I \left( \mathbf{X}_d; \mathbf{Y}_d \;\middle|\; \hat{\mathbf{H}} \right) \notag \\
			&\overset{(a)}{\ge} \frac{T - T_\tau}{T} \mathbb{E}_{\hat{\mathbf{H}}} \left[ \log_2 \det \left( \mathbf{I}_N + \mathbf{R}_{Z_{\text{eff}}}^{-1} \hat{\mathbf{H}} \mathbf{R}_{X_d} \hat{\mathbf{H}}^H \right) \right] \notag \\
			&= \frac{T - T_\tau}{T} \mathbb{E}_{\hat{\mathbf{H}}} \left[ \log_2 \det \left( \mathbf{I}_N + \frac{1}{P_d\sigma_{\tilde{H}}^2 + \sigma_z^2} \hat{\mathbf{H}} \left(\frac{P_d}{M}\mathbf{I}_M\right) \hat{\mathbf{H}}^H \right) \right] \notag \\
			&\overset{(b)}{=} \frac{T - T_\tau}{T} \mathbb{E}_{\bar{\mathbf{H}}} \left[ \log_2 \det \left( \mathbf{I}_N + \rho_{\text{eff}} \frac{\bar{\mathbf{H}}\bar{\mathbf{H}}^H}{M} \right) \right], \label{eq:rate_final}
		\end{align}
		where inequality (a) leverages the worst-case Gaussian noise bound; in equality (b), $\bar{\mathbf{H}} := \hat{\mathbf{H}}/\sigma_{\hat{H}}$ is the normalized channel estimate matrix with its entries being i.i.d. $\sim\mathcal{CN}(0,1)$, $\sigma_{\tilde{H}}^2 = (1 + \rho_\tau T_\tau/M)^{-1}$ is the variance of the LMMSE estimation error, and $\sigma_{\hat{H}}^2 = 1 - \sigma_{\tilde{H}}^2$. The effective SNR $\rho_{\mathrm{eff}}$ is defined as
		\begin{align}
			\rho_{\text{eff}} := \frac{\rho_d \rho_\tau T_\tau / M}{1 + \rho_d + \rho_\tau T_\tau / M}. \label{eq:rho_eff}
		\end{align}
		
		For the sake of simplicity, this achievable lower bound will be referred to as the ergodic communication rate in the remainder of this paper.

		\section{Pilot Sensing Scheme}\label{sec:pilot_sensing}
		
		In this section, we consider a baseline scheme that utilizes only pilots for sensing. Under a given sensing distortion constraint and power budget constraint, we aim to maximize the communication rate by jointly optimizing the pilot duration and the power allocation between pilots and data. Specifically, we first formulate the optimization problem, and subsequently derive the $R(D)$ functions for two distinct scenarios: optimal power allocation and equal power allocation.
		
		\subsection{Problem Formulation}
		
		In the PS scheme, the Rx utilizes only the observed signal $\mathbf{Y}_\tau$ to estimate $\mathbf{H}$ for target sensing. According to Proposition \ref{prop:optimal_pilot} in Section \ref{subsec:performance_metric}, the minimum sensing MSE achieved by employing the optimal orthogonal pilots is
		\begin{align}
			\mathrm{MSE} = \frac{NM}{1 + \rho_\tau T_\tau / M}.
		\end{align}
		Subject to the sensing distortion constraint $D$ and the total power budget, the optimization problem is formulated as follows\footnote{To facilitate subsequent derivations, this paper expresses the transmit power allocation in terms of the normalized quantities $\rho_\tau = P_\tau/\sigma_z^2$, $\rho_d = P_d/\sigma_z^2$, and $\rho = P/\sigma_z^2$.}:
		\begin{subequations}\label{eq:P1}
			\begin{align}
				(\mathcal{P}1): \max_{T_\tau, \rho_\tau, \rho_d} \quad 
				&  \frac{T - T_\tau}{T} \mathbb{E} \left[ \log_2 \det \left( \mathbf{I}_N 
				+ \frac{\rho_{\text{eff}}}{M} \bar{\mathbf{H}}\bar{\mathbf{H}}^H \right) \right] 
				\label{eq:obj_P1} \\
				\text{s.t.} \quad 
				& \frac{NM}{1 + \rho_\tau T_\tau / M} \le D, \label{eq:P1_sensing} \\
				& \rho_\tau T_\tau + \rho_d(T - T_\tau) \le \rho T, \label{eq:P1_energy} \\
				& T_\tau \ge M, \quad \rho_\tau \ge 0, \quad \rho_d \ge 0, \label{eq:cons_boundary}
			\end{align}
		\end{subequations}
		where $\rho_{\mathrm{eff}}$ is defined in \eqref{eq:rho_eff}. Here, \eqref{eq:P1_sensing} and \eqref{eq:P1_energy} represent the sensing distortion constraint and the total power budget constraint, respectively.
		
		\subsection{$R(D)$ with Optimal Power Allocation}
		
		For problem $(\mathcal{P}1)$, when the pilot and data powers are optimized separately, the optimal pilot duration admits the following concise analytical characterization.
		\begin{proposition}\label{prop:optimal_time}
			When the pilot and data powers are optimized separately, the optimal pilot duration satisfies $T_\tau^* = M$.
		\end{proposition}
		\begin{IEEEproof}
			Please refer to Appendix \ref{app:B} .
		\end{IEEEproof}
		\begin{remark}
			The sensing accuracy depends solely on the total power expended during the pilot phase, denoted as $\rho_\tau T_\tau$. In contrast, the communication rate scales linearly with respect to time, but only logarithmically with respect to power. Therefore, under the premise of satisfying the sensing constraint, the system should compress the pilot duration to its minimum allowable value, $M$, and compensate for the required sensing power by increasing $\rho_\tau$, thereby allocating more temporal resources to data transmission.
		\end{remark}
		\begin{remark}
			When the sensing distortion is minimized to its theoretical lower bound $D = \frac{NM^2}{\rho T + M}$, which corresponds to allocating all available power to the pilots, the communication rate becomes zero. In this case, the problem degrades to pure channel estimation. As long as the total pilot power remains constant, any configuration satisfying $T_\tau \geq M$ can achieve the same minimum sensing distortion. Consequently, $T_\tau = M$ is not the unique optimal solution in this scenario.
		\end{remark}
		
		By substituting Proposition \ref{prop:optimal_time} into problem $(\mathcal{P}1)$ and setting $T_\tau = M$, the original problem can be simplified into a power allocation problem between pilots and data, formulated as follows:
		\begin{subequations}\label{eq:P2}
			\begin{align}
				(\mathcal{P}2): \max_{\rho_\tau, \rho_d} \quad 
				&  \frac{T - M}{T} \mathbb{E} \left[ \log_2 \det \left( \mathbf{I}_N 
				+ \frac{\rho_{\mathrm{eff}}}{M} \bar{\mathbf{H}}\bar{\mathbf{H}}^H \right) \right] 
				\label{eq:obj_P2} \\
				\text{s.t.} \quad 
				& \frac{NM}{1 + \rho_\tau} \le D, \label{eq:P2_sensing} \\
				& \rho_\tau M + \rho_d (T - M) \le \rho T, \label{eq:P2_energy} \\
				& \rho_\tau \ge 0, \quad \rho_d \ge 0, \label{eq:P2_boundary}
			\end{align}
		\end{subequations}
		where
		\begin{align}
			\rho_{\mathrm{eff}}
			= \frac{\rho_d \rho_\tau}{1 + \rho_d + \rho_\tau}.
			\label{eq:rho_eff_P2}
		\end{align}
		Since the objective function is monotonically increasing with respect to $\rho_{\mathrm{eff}}$, and $\rho_{\mathrm{eff}}$ is monotonically increasing with respect to both $\rho_\tau$ and $\rho_d$, the constraints \eqref{eq:P2_sensing} and \eqref{eq:P2_energy} must hold with equality at the optimal solution. As a result, the two power variables are uniquely determined by the distortion $D$:
		\begin{align}\label{eq:power_D}
			\rho_\tau(D) = \frac{NM-D}{D}, \qquad \rho_d(D) = \frac{\rho T - M \rho_\tau(D)}{T-M}.
		\end{align}
		Substituting \eqref{eq:power_D} into \eqref{eq:rho_eff_P2} yields the $R(D)$ function for the PS scheme. 
		By analyzing the concavity and extrema of $R(D)$, we establish the following theorem.
		\begin{theorem}\label{thm:Rate_Distortion}
			Under optimal power allocation, the $R(D)$ function for the PS scheme is given by
			\begin{align}
				R(D)
				= \frac{T - M}{T}\,
				\mathbb{E}\!\left[
				\log_2\det\!\left(
				\mathbf{I}_N
				+ \frac{\rho_{\mathrm{eff}}(D)}{M}
				\bar{\mathbf{H}}\bar{\mathbf{H}}^H
				\right)
				\right],
				\label{eq:CD_pilot}
			\end{align}
			where
			\begin{align}
				\rho_{\mathrm{eff}}(D)
				= \frac{(NM - D)\!\left[D(\rho T + M) - NM^2\right]}
				{D\!\left[D(\rho T + M) + NM(T - 2M)\right]}.
				\label{eq:rho_eff_D}
			\end{align}
			The effective domain of $R(D)$ is $D \in [D_{\min}^{(1)},\, D^{(1)*}]$. Within this interval, $R(D)$ is strictly monotonically increasing and concave with respect to $D$, where
			\begin{align}
				D_{\min}^{(1)} &= \frac{NM^2}{\rho T + M}, \label{eq:D_min}\\
				D^{(1)*} &= \frac{NM^2}{T(1+\rho)}\left[1 + \sqrt{\frac{(T(1+\rho)-M)(T-M)}{M(\rho T+M)}}\right]. \label{eq:D_star}
			\end{align}
		\end{theorem}
		\begin{IEEEproof}
			please refer to Appendix \ref{app:C}.
		\end{IEEEproof}
		\begin{remark}
			When all power is allocated to the pilots, the minimum sensing distortion $D_{\min}^{(1)}=\frac{NM^2}{\rho T+M}$ is achieved. Since no power remains for data transmission, the corresponding communication rate is zero. On the other hand, allocating zero power to the pilots yields the maximum sensing distortion $D_{\max}^{(1)}=NM$.
			In this case, the channel estimate
			$\hat{\mathbf{H}}$ degenerates to an all-zero matrix, and consequently the achievable rate in \eqref{eq:rate_final} also becomes zero. Theorem \ref{thm:Rate_Distortion} indicates that the operationally meaningful range of the sensing distortion $D$ is  $[D_{\min}^{(1)},\, D^{(1)*}]$. Targeting a  distortion level $D > D^{(1)*}$ by intentionally lowering the pilot power leads to a severe deterioration in channel estimation accuracy, while the resulting increase in data transmission power is insufficient to compensate for this loss, ultimately reducing the communication rate. In other words, the distortion constraint becomes inactive when $D > D^{(1)*}$. In this regime, it is optimal to operate at the sensing distortion level $D^{(1)*}$;
			further relaxing the distortion constraint provides no additional communication-rate gain.
		\end{remark}
		
		

		\subsection{$R(D)$ with Equal Power Allocation}
		
		Under  equal power allocation, i.e., $\rho_\tau = \rho_d = \rho$, the total power constraint \eqref{eq:P1_energy} is automatically satisfied. The only remaining optimization variable is the pilot duration $T_\tau$. Thus, problem $(\mathcal{P}1)$ simplifies to:
		\begin{subequations}\label{eq:P3}
			\begin{align}
				(\mathcal{P}3): \max_{T_\tau} \quad 
				&  \frac{T - T_\tau}{T} \mathbb{E} \left[ \log_2 \det \left( \mathbf{I}_N + \frac{\rho_{\text{eff}}(T_\tau)}{M} \bar{\mathbf{H}}\bar{\mathbf{H}}^H \right) \right] \label{eq:obj_P3} \\
				\text{s.t.} \quad 
				& \frac{NM}{1 + \rho T_\tau / M} \le D, \label{eq:P3_sensing} \\
				& M \le T_\tau \le T, \label{eq:P3_time_bound}
			\end{align}
		\end{subequations}
		where 
		\begin{align}
			\rho_{\text{eff}}(T_\tau) = \frac{\rho^2 T_\tau / M}{1 + \rho + \rho T_\tau / M}. \label{eq:rho_eff_Ttau_equal}
		\end{align}
		By combining constraints \eqref{eq:P3_sensing} and \eqref{eq:P3_time_bound}, the feasible region for the sensing distortion $D$ is given by
		\begin{align}
			D \in \left[ \frac{NM^2}{\rho T + M}, \quad \frac{NM}{1 + \rho} \right]. \label{eq:D_range_equal_power}
		\end{align}
		where the minimum distortion $D_{\min}^{(2)}$ corresponds to $T_\tau = T$, meaning that all symbols are used for pilots, resulting in a zero communication rate. Conversely, $D_{\max}^{(2)}$ corresponds to $T_\tau = M$, where the pilot duration takes its minimum allowable value and the communication rate is maximized.
		
		From constraint \eqref{eq:P3_sensing}, $T_\tau$ can be explicitly expressed as a function of $D$:
		\begin{equation}
			T_\tau(D)  = \frac{M(NM - D)}{\rho D}. \label{eq:T_tau_D}
		\end{equation}
		By substituting \eqref{eq:T_tau_D} into \eqref{eq:rho_eff_Ttau_equal}, the closed-form expression for the effective SNR with respect to $D$ is obtained as
		\begin{align}
			\rho_{\text{eff}}(D) 
			= \frac{\rho (NM - D)}{NM + \rho D}. \label{eq:rho_eff_D_equal_power}
		\end{align}
		Based on the above results, the $R(D)$ function under equal power allocation is presented below.
		\begin{theorem}\label{thm:Rate_Distortion_equal_power}
			Under equal power allocation, the $R(D)$ function for the PS scheme is given by
			\begin{align}
				R(D) &= \frac{D(\rho T + M) - NM^2}{\rho T D} \nonumber \\
				&\quad \mathbb{E}\!\left[ \log_2 \det \left( \mathbf{I}_N 
				+ \frac{\rho(NM - D)}{M(NM + \rho D)} \bar{\mathbf{H}}\bar{\mathbf{H}}^H \right) \right]. \label{eq:RD_equal_power}
			\end{align}
			Within the feasible distortion interval $D \in [D_{\min}^{(2)}, D_{\max}^{(2)}]$, there exists a unique $D^{(2)*}$ such that $R(D)$ is a concave function on $[D_{\min}^{(2)}, D^{(2)*}]$, and $R(D^{(2)*}) = \max_{D \in [D_{\min}^{(2)},\, D_{\max}^{(2)}]} R(D)$.
		\end{theorem}
		\begin{IEEEproof}
			please refer to Appendix \ref{app:D}.
		\end{IEEEproof}
		\begin{remark}
			Under the equal power constraint, the stationary point condition for $D^{(2)*}$ (or equivalently, $T_\tau^*$) involves an expectation term that is analytically intractable to simplify. In general, no closed-form expression exists, and the optimal solution is typically determined through a one-dimensional numerical search.
		\end{remark}

		\section{Data-Aided Sensing Scheme}\label{sec:data_assisted_sensing}

		In the DAS scheme, the Rx jointly utilizes the pilot matrix $\mathbf{X}_\tau$ and the data matrix $\mathbf{X}_d$ for sensing. The complete transmitted signal matrix is given by $\mathbf{X} = [\mathbf{X}_\tau,\, \mathbf{X}_d]$. The corresponding sensing MSE is 
		\begin{align}
			\operatorname{MSE} 
			&= \mathbb{E}_{\mathbf{X}_d}\left[ \operatorname{Tr}\left( \left( \mathbf{I}_{M} + \frac{1}{\sigma_z^2} (\mathbf{X}\mathbf{X}^H)^* \right)^{-1} \otimes \mathbf{I}_N \right) \right] \notag \\
			&= N \cdot \mathbb{E}_{\mathbf{X}_d}\left[ \operatorname{Tr}\left( \left( \mathbf{I}_M + \frac{1}{\sigma_z^2} (\mathbf{X}_\tau \mathbf{X}_\tau^H + \mathbf{X}_d \mathbf{X}_d^H)^* \right)^{-1} \right) \right] \notag \\
			&\overset{(a)}{=} N \cdot \mathbb{E}_{\mathbf{X}_d}\left[ \operatorname{Tr}\left( \left( \underbrace{\left( 1 + \frac{\rho_\tau T_\tau}{M} \right)}_{ c} \mathbf{I}_M + \frac{\mathbf{X}_d \mathbf{X}_d^H }{\sigma_z^2} \right)^{-1} \right) \right], \label{eq:MSE_average}
		\end{align}
		where $(a)$ assumes that the pilot sequences are orthogonal, and the constant $c := 1 + \rho_\tau T_\tau/M$ characterizes the prior gain provided by the pilots.
		
		It should be pointed out that, due to the inherent randomness of the data matrix $\mathbf{X}_d$, the expectation in \eqref{eq:MSE_average} involves a nonlinear matrix inversion, which makes direct calculation rather difficult. To address this, we first derive a theoretical lower bound for the MSE of the DAS scheme in this section, and subsequently employ RMT to derive its accurate asymptotic analytical expression.

		\subsection{Lower Bound for Sensing Distortion}
		
		Based on Jensen's inequality, a theoretical lower bound on the MSE of DAS can be obtained, as stated in the following proposition.
		\begin{proposition}\label{prop:lower_bound}
			For any finite dimensions $M$ and $T_d$, the sensing MSE of the DAS scheme satisfies:
			\begin{align}
				\operatorname{MSE} \ge  \frac{NM^2}{M + \rho T}. \label{eq:MSE_lower_bound}
			\end{align}
		\end{proposition}
		\begin{IEEEproof}
			Please refer to Appendix \ref{app:E} .
		\end{IEEEproof}
		\begin{remark}
			This lower bound essentially idealizes the random data matrix as a deterministic, orthogonal, and full-rank matrix, i.e., $\frac{1}{T_d}\mathbf{X}_d \mathbf{X}_d^H = \frac{P_d}{M}\mathbf{I}_M$. The inevitable random fluctuations inherent in finite-length data sequences will result in sensing performance degradation. Consequently, this lower bound is generally not sufficiently tight and therefore cannot accurately guide resource allocation design in practical ISAC systems.
		\end{remark}

		\subsection{Asymptotic Expression for Sensing Distortion}
		
		RMT provides a powerful mathematical framework for analyzing the eigenvalue distribution of random matrices, and has been widely applied in various fields such as wireless communications \cite{tulino2004random}, quantum physics \cite{guhr1998random}, and statistical inference \cite{johnstone2006high}. Therefore, in this section, we resort to RMT to derive a closed-form asymptotic expression for the MSE of the DAS scheme.
		
		Let $\mathbf{X}_d = \sqrt{P_d/M}\,\mathbf{G}$, where the entries of $\mathbf{G} \in \mathbb{C}^{M \times T_d}$ are i.i.d. $\sim \mathcal{CN}(0, 1)$. Then, from \eqref{eq:MSE_average}, we have
		\begin{align}
			\operatorname{MSE} &= N \cdot \mathbb{E}_{\mathbf{G}}\left[ \operatorname{Tr}\left( \left(c\mathbf{I}_M + \frac{\rho_d}{M}\mathbf{G}\mathbf{G}^H\right)^{-1} \right) \right]. \label{eq:MSE_G}
		\end{align}
		Define $\mathbf{W} := \frac{1}{M}\mathbf{G}\mathbf{G}^H$ as a standard complex Wishart matrix. Its Stieltjes transform is defined as
		\begin{align}
			m_{\mathbf{W}}(z) &:= \frac{1}{M}\operatorname{Tr}\left( \left(\mathbf{W} - z\mathbf{I}_M\right)^{-1} \right), \quad z \in \mathbb{C} \setminus \mathbb{R}^+. \label{eq:Stieltjes_def}
		\end{align}
		Note that
		\[
		c\mathbf{I}_M + \frac{\rho_d}{M}\mathbf{G}\mathbf{G}^H
		= \rho_d\left(\frac{c}{\rho_d}\mathbf{I}_M + \mathbf{W}\right),
		\]
		thus the MSE can be expressed in terms of the trace of the inverse of $\mathbf{W}$. To facilitate subsequent analysis, we define the normalized trace:
		\begin{align}
			x &:= \frac{1}{M} \mathbb{E}_{\mathbf{G}}\left[ \operatorname{Tr}\left( \left(c\mathbf{I}_M + \frac{\rho_d}{M}\mathbf{G}\mathbf{G}^H\right)^{-1} \right) \right] \notag \\
			&= \frac{1}{\rho_d} \cdot \frac{1}{M} \mathbb{E}_{\mathbf{W}}\left[ \operatorname{Tr}\left( \left(\frac{c}{\rho_d}\mathbf{I}_M + \mathbf{W}\right)^{-1} \right) \right] \notag \\
			&= \frac{1}{\rho_d} m_{\mathbf{W}}\left(-\frac{c}{\rho_d}\right), \label{eq:x_def}
		\end{align}
		and the sensing MSE can be written as $\operatorname{MSE} = NMx$.
		
		As $M, T_d \to \infty$ while the ratio $\gamma := M/T_d$ remains finite, the empirical spectral distribution of $\mathbf{W}$ converges almost surely to the Marchenko-Pastur law \cite{tse1999linear}, and its Stieltjes transform almost surely satisfies
		\begin{align}
			m_{\mathbf{W}}(z) &= \frac{1}{-z + \gamma\left(1 + m_{\mathbf{W}}(z)\right)^{-1}}. \label{eq:SC_equation}
		\end{align}
		Substituting $z = -c/\rho_d$ into \eqref{eq:SC_equation} and simplifying the resulting expression, we obtain a quadratic equation with respect to $x$:
		\begin{align}
			c\rho_d x^2 + b x - 1 &= 0,  \label{eq:quadratic}
		\end{align}
		where $b := c + \rho_d T_d/M - \rho_d$. Since $x$ corresponds to the trace of a positive definite matrix, taking the positive root yields
		\begin{align}
			x^* &= \frac{-b + \sqrt{b^2 + 4c\rho_d}}{2c\rho_d}. \label{eq:x_star}
		\end{align}
		Consequently, we arrive at the following proposition.
		\begin{proposition}\label{prop:MSE_asymp}
			As $M, T_d \to \infty$ with $\gamma = M/T_d$, the MSE of the DAS scheme converges almost surely to
			\begin{align}
				\frac{NM}{2c\rho_d} \left[\sqrt{\left(c + \frac{\rho_d T_d}{M} - \rho_d\right)^2 + 4c\rho_d} - \left(c + \frac{\rho_d T_d}{M} - \rho_d\right)\right], \label{eq:MSE_asymp}
			\end{align}
			where $c = 1 + \rho_\tau T_\tau/M $.
		\end{proposition}
		\begin{remark}
			Unlike the Jensen's lower bound derived in Proposition \ref{prop:lower_bound}, \eqref{eq:MSE_asymp} fully captures the spectral fluctuation effects of $\mathbf{X}_d$. Specifically, when $T_d$ is finite, $\sqrt{b^2 + 4c\rho_d} > b$, implying that the asymptotic MSE is strictly greater than the lower bound. The difference between the two precisely quantifies the sensing performance degradation introduced by the randomness. Conversely, when $T_d \to \infty$, $\gamma \to 0$, leading to $b \approx \rho_d / \gamma \to +\infty$. By performing a first-order Taylor expansion on the radical term in \eqref{eq:MSE_asymp}, i.e., $\sqrt{b^2 + 4c\rho_d} = b + 2c\rho_d/b + \mathcal{O}(1/b^2)$, and substituting it back, the asymptotic MSE is approximated as $\approx NM^2/(\rho_d T_d) \approx NM^2/(M+\rho T)$, which degrades exactly to the Jensen's lower bound.
		\end{remark}

		\subsection{Problem Formulation}
		
		According to Proposition \ref{prop:MSE_asymp}, the asymptotic MSE of the DAS scheme can be characterized by the following closed-form expression:
		\begin{equation}
			D_{\mathrm{RMT}} = \frac{NM}{2c\rho_d} 
			\Bigl[ \sqrt{b^{2} + 4c\rho_d} - b \Bigr].
			\label{eq:D_RMT_def}
		\end{equation}
		Subject to the sensing distortion constraint $D$ and the total power budget constraint, the optimization problem is formulated as follows:
		\begin{subequations}\label{eq:P_DA}
			\begin{align}
				(\mathcal{P}4): \max_{T_\tau, \rho_\tau, \rho_d} \quad 
				&  \frac{T - T_\tau}{T} \mathbb{E} \left[ \log_2 \det \left( \mathbf{I}_N + \frac{\rho_{\text{eff}}}{M} \bar{\mathbf{H}}\bar{\mathbf{H}}^H \right) \right] \label{eq:obj_P_DA} \\
				\text{s.t.} \quad 
				&  D_{\mathrm{RMT}} \le D, \label{eq:cons_DA_sensing} \\
				&  \rho_\tau T_\tau + \rho_d(T - T_\tau) \le \rho T, \label{eq:cons_DA_energy} \\
				&  T_\tau \ge M, \quad \rho_\tau \ge 0, \quad \rho_d \ge 0, \label{eq:cons_DA_boundary}
			\end{align}
		\end{subequations}
		where $\rho_{\mathrm{eff}}$ is defined in \eqref{eq:rho_eff}. Here, constraint \eqref{eq:cons_DA_sensing} represents the sensing distortion constraint, and constraint \eqref{eq:cons_DA_energy} denotes the total power budget.

		\subsection{$R(D)$ with Optimal Power Allocation}
		
		For problem $(\mathcal{P}4)$, when the pilot and data powers are optimized separately, the optimal pilot duration admits the following analytical characterization.
		\begin{proposition}\label{prop:DA_optimal_time}
			When the pilot and data powers are optimized separately, the optimal pilot duration satisfies $T_\tau^* = M$.
		\end{proposition}
		\begin{IEEEproof}
			Please refer to Appendix \ref{app:F} .
		\end{IEEEproof}
		
		Substituting the conclusion $T_\tau^* = M$ from Proposition \ref{prop:DA_optimal_time} into problem $(\mathcal{P}4)$, the original problem is simplified into a power allocation problem between pilots and data, formulated as:
		\begin{subequations}\label{eq:P_DA_reduced}
			\begin{align}
				(\mathcal{P}5): \max_{\rho_\tau, \rho_d} \quad 
				&  \frac{T - M}{T} \mathbb{E} \left[ \log_2 \det \left( \mathbf{I}_N + \frac{\rho_{\text{eff}}}{M} \bar{\mathbf{H}}\bar{\mathbf{H}}^H \right) \right] \label{eq:obj_P_DA_reduced} \\
				\text{s.t.} \quad 
				&  D_{\mathrm{RMT}} \le D, \label{eq:cons_DA_sensing_reduced} \\
				&  \rho_\tau M + \rho_d(T - M) \le \rho T, \label{eq:cons_DA_energy_reduced} \\
				&  \rho_\tau \ge 0, \quad \rho_d \ge 0, \label{eq:cons_DA_boundary_reduced}
			\end{align}
		\end{subequations}
		where
		\begin{align}
			\rho_{\text{eff}} &= \frac{\rho_\tau \rho_d}{1 + \rho_\tau + \rho_d}. \label{eq:rho_eff_reduced}
		\end{align}
		\begin{align}
			D_{\mathrm{RMT}} &= \frac{NM}{2(1 + \rho_\tau)\rho_d} \left[ \sqrt{ b^2 + 4(1 + \rho_\tau)\rho_d } - b \right], \label{eq:D_RMT_reduced} 
		\end{align}
		Since the objective function is monotonically increasing with respect to $\rho_{\text{eff}}$, constraint \eqref{eq:cons_DA_energy_reduced} must hold with equality at the optimal solution. Thus, $\rho_\tau$ can be expressed as a function of $\rho_d$:
		\begin{align}
			\rho_\tau(\rho_d) = \frac{\rho T - \rho_d(T - M)}{M}.
			\label{eq:rho_tau_rho_d}
		\end{align}
		By substituting \eqref{eq:rho_tau_rho_d} into the sensing constraint $D_{\mathrm{RMT}} = D$, the optimal data power $\rho_d(D)$ is uniquely determined, which in turn yields the $R(D)$ function. By performing concavity and extrema analysis on $R(D)$, we establish the following theorem.
		\begin{theorem}\label{thm:Rate_Distortion_DA}
			Under optimal power allocation, the $R(D)$ function for the DAS scheme is given by
			\begin{align}
				R(D) = \frac{T - M}{T}\,\mathbb{E}\!\left[
				\log_2\det\!\left(
				\mathbf{I}_N + \frac{\rho_{\mathrm{eff}}\!\left(\rho_d(D)\right)}{M}
				\bar{\mathbf{H}}\bar{\mathbf{H}}^H
				\right)
				\right],
				\label{eq:RD_DA}
			\end{align}
			where $\rho_{\mathrm{eff}}$ is given by \eqref{eq:rho_eff_reduced}, and the optimal data power $\rho_d(D)$ admits a unique closed-form expression:
			\begin{align}
				\rho_d(D) = \frac{M \left( (u - \kappa) + \sqrt{ (u - \kappa)^2 + \dfrac{4(T-M)}{M} \kappa (u - \kappa) } \right)}{2(T-M)},
				\label{eq:rho_d_D}
			\end{align}
			with $u := 1 + {\rho T}/{M}$ and $\kappa := {NM}/{D}$. The effective domain of $R(D)$ is $D \in [D_{\min}^{(3)},\, D^{(3)*}]$. Within this interval, $R(D)$ is strictly monotonically increasing and concave with respect to $D$, where
			\begin{align}
				D_{\min}^{(3)} &= \frac{NM^2}{M + \rho T},
				\label{eq:DA_Dmin} \\
				D^{(3)*} &= D_{\mathrm{RMT}}\!\left(\rho_d^*\right), \qquad
				\label{eq:DA_Dstar} \\
				\rho_d^* &= \frac{-(M + \rho T)}{2M - T} \nonumber \\
				&\quad + \frac{\sqrt{(M + \rho T)^2 + \dfrac{(M+\rho T)(2M-T)\rho T}{T-M}}}{2M - T}. \notag
			\end{align}
		\end{theorem}
		\begin{IEEEproof}
			please refer to Appendix \ref{app:G}.
		\end{IEEEproof}
		\begin{remark}
			The minimum distortion $D_{\min}^{(3)} $ in \eqref{eq:DA_Dmin} is identical to that of the PS scheme in \eqref{eq:D_min}. In the limit as $\rho_d \to 0$, the RMT asymptotic bound of the DAS scheme reduces to Jensen's lower bound. Furthermore, $D^{(3)*}$ corresponds to the distortion value that achieves the maximum communication rate.
		\end{remark}

		\subsection{$R(D)$ with Equal Power Allocation}
		
		Under  equal power allocation, i.e., $\rho_\tau = \rho_d = \rho$, the total power constraint is automatically satisfied. The only remaining optimization variable is the pilot duration $T_\tau$. Thus, problem $(\mathcal{P}4)$ simplifies to:
		\begin{subequations}\label{prob:FP}
			\begin{align}
				(\mathcal{P}6): \max_{T_\tau} \quad 
				& 
				\frac{T - T_\tau}{T}\,\mathbb{E}\!\left[
				\log_2\det\!\left(
				\mathbf{I}_N + \frac{\rho_{\mathrm{eff}}(T_\tau)}{M}\bar{\mathbf{H}}\bar{\mathbf{H}}^H
				\right)
				\right] \label{prob:FP_obj} \\
				\text{s.t.} \quad
				& D_{\mathrm{RMT}} \leq D, \label{prob:FP_c1} \\
				& M \leq T_\tau \leq T, \label{prob:FP_c2}
			\end{align}
		\end{subequations}
		where
		\begin{align}
			\rho_{\mathrm{eff}}(T_\tau) = \frac{\rho^2 T_\tau/M}{1 + \rho + \rho T_\tau/M},
			\label{eq:rho_eff_FP}
		\end{align}
		\begin{align}
			D_{\mathrm{RMT}}(T_\tau) = \frac{2NM}{\sqrt{b^2 + 4\rho c} + b}.
			\label{eq:D_FP_simplified}
		\end{align}
		with $c = 1 + \rho T_\tau/M$ and $b = 1 + \rho T/M - \rho$.
		
		Since $c$ is strictly monotonically increasing with respect to $T_\tau$, $D_{\mathrm{RMT}}(T_\tau)$ is strictly monotonically decreasing with respect to $T_\tau$. Therefore, the equation $D_{\mathrm{RMT}}(T_\tau) = D$ has a unique solution $T_\tau(D)$, which leads to the following theorem.
		\begin{theorem}\label{thm:RD_FP}
			Under the equal power constraint, the $R(D)$ function for the DAS scheme is given by
			\begin{align}
				R(D) = \frac{T - T_\tau(D)}{T}\,\mathbb{E}\!\left[
				\log_2\det\!\left(
				\mathbf{I}_N + \frac{\rho_{\mathrm{eff}}\!\left(T_\tau(D)\right)}{M}
				\bar{\mathbf{H}}\bar{\mathbf{H}}^H
				\right)
				\right],
				\label{eq:RD_FP}
			\end{align}
			where $T_\tau(D)$ admits the closed-form expression:
			\begin{align}
				T_\tau(D) = \frac{M}{\rho}\!\left[\frac{\kappa(\kappa - b)}{\rho} - 1\right].
				\label{eq:Ttau_D_FP}
			\end{align}
			The effective domain of $R(D)$ is $D \in [D_{\min}^{(4)},\, D^{(4)*}]$. Within this interval, $R(D)$ is strictly concave, where
			\begin{align}
				D_{\min}^{(4)} &= \frac{NM^2}{M + \rho T},
				\label{eq:FP_Dmin} \\
				D^{(4)*} &= \frac{2NM}{\sqrt{b^2
						+ 4\rho\!\left(1 + \dfrac{\rho T_\tau^*}{M}\right)} + b},
				\label{eq:FP_Dstar}
			\end{align}
			and $T_\tau^*$ is the unique stationary point of the objective function \eqref{prob:FP_obj} satisfying $\frac{dR}{dT_\tau} = 0$.
		\end{theorem}
		\begin{proof}
			The proof for the strict concavity of $R(T_\tau)$ and the existence of the unique extremum point $T_\tau^*$ is identical to that in Appendix \ref{app:D}, and is thus omitted here. 
		\end{proof}
		\begin{remark}
			In general, there is no closed-form expression for $D^{(4)*}$ (or equivalently, $T_\tau^*$), since its stationary point condition involves the expectation of a matrix log-determinant, which is analytically intractable. Nevertheless, thanks to the global strict concavity of $R(T_\tau)$ over $[M, T]$, $T_\tau^*$ can be efficiently determined via a one-dimensional numerical search.
		\end{remark}

		\section{Asymptotic Performance Analysis}\label{sec:asymptotic_analysis}
		While it is intuitively expected that incorporating data payloads yields better sensing performance, our asymptotic analysis rigorously quantifies the exact magnitudes of these performance gains.
		To compare the performance boundaries of the PS and DAS schemes within a unified framework, we map the sensing distortion $D$ to the effective SNR\footnote{This mapping stems from converting the total sensing distortion $D$ into the average estimation error variance of a single channel coefficient, $\sigma_{\tilde{H}}^2 = D/(NM)$, and substituting it into the standard form of the effective SNR under imperfect CSI.}:
		\begin{align}
			\rho_{\mathrm{eff}}
			= \frac{\rho_d(NM - D)}{NM + \rho_d D}.
			\label{eq:rho_eff_mapping}
		\end{align}
		Subsequently, we analytically evaluate the performance of each scheme under two asymptotic limits: the low-SNR regime ($\rho \to 0$) and the high-SNR regime ($\rho \to \infty$), respectively.

		\subsection{Low-SNR Asymptotic Analysis}
		
		As $\rho \to 0$, both $\rho_\tau \to 0$ and $\rho_d \to 0$. Define $x_1 := \rho_\tau T_\tau/M$ and $x_2 := \rho_d T_d/M$; it then follows that $x_1, x_2 \to 0$.
		
		For the PS scheme, the first-order approximation of the sensing distortion is given by
		\begin{align}
			D_{\mathrm{PS}}
			&= \frac{NM}{1 + x_1} \nonumber \\
			&\approx NM(1 - x_1) \nonumber \\
			&= NM - N\rho_\tau T_\tau.
			\label{eq:D_Trad_lowSNR}
		\end{align}
		For the DAS scheme, noting that $b = 1 + x_1 + x_2 - \rho_d$, we perform a first-order Taylor expansion on the radical term as $\sqrt{b^2 + 4(1+x_1)\rho_d} \approx b + \frac{2(1+x_1)\rho_d}{b}$. By simplifying, we obtain
		\begin{align}
			D_{\mathrm{DAS}}
			&\approx \frac{NM}{1 + x_1 + x_2} \nonumber \\
			&\approx NM - N\rho T.
			\label{eq:D_RMT_lowSNR}
		\end{align}
		Substituting the sensing distortions \eqref{eq:D_Trad_lowSNR} and \eqref{eq:D_RMT_lowSNR} for both schemes into \eqref{eq:rho_eff_mapping}, respectively, yields
		\begin{subequations}\label{eq:lowSNR_approx_rho_eff}
			\begin{align}
				\rho_{\mathrm{eff}}^{\mathrm{PS}}
				&\approx \frac{\rho_d(NM - D_{\mathrm{PS}})}{NM}
				\approx \rho_d \rho_\tau \frac{T_\tau}{M},
				\label{eq:rho_eff_Trad_approx} \\[4pt]
				\rho_{\mathrm{eff}}^{\mathrm{DAS}}
				&\approx \frac{\rho_d(NM - D_{\mathrm{DAS}})}{NM}
				\approx \rho_d \frac{\rho T}{M}.
				\label{eq:rho_eff_RMT_approx}
			\end{align}
		\end{subequations}
		In the following, we analyze the asymptotic gains of each scheme under both optimal power allocation and equal power allocation strategies.
		
		\subsubsection{Optimal Power Allocation}
		
		From Proposition \ref{prop:DA_optimal_time}, we have $T_\tau^* = M$, which implies $T_d = T-M$. Under the total energy constraint $\rho_\tau M + \rho_d(T-M) = \rho T$, maximizing $\rho_{\mathrm{eff}}^{\mathrm{PS}} \propto \rho_d\rho_\tau$ is equivalent to maximizing the product $\rho_\tau M \cdot \rho_d(T-M)$. According to the arithmetic–geometric mean inequality, the optimal solution is attained when the two terms are equal:
		\begin{align}
			\rho_\tau = \frac{\rho T}{2M}, \qquad
			\rho_d    = \frac{\rho T}{2(T-M)}.
		\end{align}
		Substituting these into \eqref{eq:lowSNR_approx_rho_eff} yields
		\begin{subequations}
			\begin{align}
				\rho_{\mathrm{eff}}^{\mathrm{PS}}
				&\approx \frac{T^2}{4M(T-M)}\rho^2, \\
				\rho_{\mathrm{eff}}^{\mathrm{DAS}}
				&\approx \frac{T^2}{2M(T-M)}\rho^2.
			\end{align}
		\end{subequations}
		Therefore, the DAS scheme achieves a strict 3 dB low-SNR gain over the PS scheme:
		\begin{align}
			\lim_{\rho \to 0}
			\frac{\rho_{\mathrm{eff}}^{\mathrm{DAS}}}{\rho_{\mathrm{eff}}^{\mathrm{PS}}} = 2.
		\end{align}

		\subsubsection{Equal Power Allocation}
		
		Under the condition of $\rho_\tau = \rho_d = \rho$, substituting this into \eqref{eq:lowSNR_approx_rho_eff} yields
		\begin{subequations}\label{eq:lowSNR_fixed_rho_eff}
			\begin{align}
				\rho_{\mathrm{eff}}^{\mathrm{PS}}
				&\approx \frac{\rho^2 T_\tau}{M},
				\label{eq:rho_eff_Trad_fixed} \\
				\rho_{\mathrm{eff}}^{\mathrm{DAS}}
				&\approx \frac{\rho^2 T}{M}.
				\label{eq:rho_eff_RMT_fixed}
			\end{align}
		\end{subequations}
		In the low-SNR regime, the communication rate is proportional to $\frac{T-T_\tau}{T} \rho_{\mathrm{eff}}$. For the PS scheme, maximizing $(T-T_\tau)T_\tau$ subject to the constraint $T_\tau \geq M$ gives $T_\tau^* = \max(M,\, T/2)$. For the DAS scheme, $\rho_{\mathrm{eff}}^{\mathrm{DAS}}$ is independent of $T_\tau$, implying that any $T_\tau \geq M$ can achieve the same asymptotic performance. By substituting the optimal $T_\tau^*$ into \eqref{eq:lowSNR_fixed_rho_eff}, the optimal effective SNRs are obtained as:
		\begin{subequations}
			\begin{align}
				\rho_{\mathrm{eff}}^{\mathrm{PS}} &=
				\begin{cases}
					\dfrac{\rho^2 T}{2M}, & T \geq 2M, \\[4pt]
					\rho^2,               & T < 2M,
				\end{cases} \\[4pt]
				\rho_{\mathrm{eff}}^{\mathrm{DAS}} &= \frac{\rho^2 T}{M}.
			\end{align}
		\end{subequations}
		Therefore, in typical scenarios with large coherence blocks ($T \geq 2M$), the DAS scheme likewise achieves a 3 dB low-SNR gain:
		\begin{align}
			\lim_{\rho \to 0}
			\frac{\rho_{\mathrm{eff}}^{\mathrm{DAS}}}{\rho_{\mathrm{eff}}^{\mathrm{PS}}} =
			\begin{cases}
				2,   & T \geq 2M, \\
				T/M, & T < 2M.
			\end{cases}
		\end{align}

		\subsection{High-SNR Asymptotic Analysis}
		
		As $\rho \to \infty$, the time-multiplexing loss becomes the primary bottleneck. Therefore, both schemes adopt the shortest pilot duration, $T_\tau^* = M$, which implies $T_d = T - M$. Let us define $K := T_d/M = (T-M)/M$.
		
		We first present the asymptotic expressions of the sensing distortion for both schemes in the high-SNR regime. For the PS scheme, as $\rho_\tau \to \infty$, we have
		\begin{align}
			D_{\mathrm{PS}} \approx \frac{NM}{\rho_\tau}.
			\label{eq:D_Trad_highSNR}
		\end{align}
		For the DAS scheme, in the high-SNR regime, $b \approx \rho_\tau + (K-1)\rho_d$. Thus, the sensing distortion can be written as
		\begin{align}
			D_{\mathrm{DAS}}
			\approx \frac{NM}{2\rho_\tau\rho_d}
			\left(\sqrt{b^2 + 4\rho_\tau\rho_d} - b\right).
			\label{eq:D_RMT_highSNR}
		\end{align}
		Substituting \eqref{eq:D_Trad_highSNR} and \eqref{eq:D_RMT_highSNR} into \eqref{eq:rho_eff_mapping}, respectively, yields
		\begin{subequations}\label{eq:highSNR_general_rho_eff}
			\begin{align}
				\rho_{\mathrm{eff}}^{\mathrm{PS}} &\approx \frac{\rho_\tau\rho_d}{\rho_\tau + \rho_d}, 
				\label{eq:rho_eff_Trad_highSNR_general} \\
				\rho_{\mathrm{eff}}^{\mathrm{DAS}} &\approx \frac{ \rho_d}{1 + \rho_d \frac{D_{\mathrm{DAS}}}{NM}}. 
				\label{eq:rho_eff_RMT_highSNR_general}
			\end{align}
		\end{subequations}
		In the following, we analyze the asymptotic gains of each scheme under both optimal power allocation and equal power allocation strategies.

		\subsubsection{Optimal Power Allocation}
		
		For the PS scheme, maximizing \eqref{eq:rho_eff_Trad_highSNR_general} subject to the total energy constraint $\rho_\tau + K\rho_d = (1+K)\rho$ yields
		\begin{align}
			\rho_{\mathrm{eff}}^{\mathrm{PS}} \approx \frac{1+K}{(1+\sqrt{K})^2}\,\rho.
			\label{eq:rho_eff_Trad_highSNR_opt}
		\end{align}
		For the DAS scheme, when the condition $4\rho_\tau\rho_d/b^2 \ll 1$ is satisfied, a first-order Taylor expansion on the radical term in \eqref{eq:D_RMT_highSNR} gives $\sqrt{b^2 + 4\rho_\tau\rho_d} \approx b + 2\rho_\tau\rho_d/b$. This simplifies to $D_{\mathrm{DAS}} \approx NM/b$. Substituting this result into \eqref{eq:rho_eff_RMT_highSNR_general} and applying the energy constraint $b + \rho_d \approx (1+K)\rho$, the effective SNR can be simplified into a quadratic function of $\rho_d$:
		\begin{align}
			\rho_{\mathrm{eff}}^{\mathrm{DAS}}
			\approx \frac{\rho_d b}{b + \rho_d}
			= \frac{\rho_d\bigl[(1+K)\rho - \rho_d\bigr]}{(1+K)\rho}.
			\label{eq:rho_eff_RMT_highSNR}
		\end{align}
		Since $\rho_\tau \geq 0$, it follows that $\rho_d \in [0,\,(1+K)\rho/K]$. Maximizing \eqref{eq:rho_eff_RMT_highSNR} within this feasible region yields
		\begin{align}
			\rho_{\mathrm{eff}}^{\mathrm{DAS}} \gtrsim
			\begin{cases}
				\dfrac{1+K}{4}\,\rho,          & 1 < K < 2, \\[6pt]
				\dfrac{K^2-1}{K^2}\,\rho,      & K \geq 2.
			\end{cases}
			\label{eq:rho_eff_RMT_highSNR_opt}
		\end{align}
		\begin{remark}
			Equation \eqref{eq:rho_eff_RMT_highSNR_opt} represents the asymptotic lower bound for the DAS scheme. In scenarios with large coherence blocks where $K \geq 2$, if $\rho_\tau$ grows at a sublinear rate (e.g., $\rho_\tau = \rho^{1-\delta}$, with $\delta \in (0,1)$), the expansion condition $4\rho_\tau\rho_d/b^2 \to 0$ holds strictly, rendering this lower bound tight.
		\end{remark}
		Based on \eqref{eq:rho_eff_Trad_highSNR_opt} and \eqref{eq:rho_eff_RMT_highSNR_opt}, the high-SNR gain ratio between the two schemes satisfies:
		\begin{align}
			\lim_{\rho\to\infty} \frac{\rho_{\mathrm{eff}}^{\mathrm{DAS}}}{\rho_{\mathrm{eff}}^{\mathrm{PS}}}
			\geq
			\begin{cases}
				\dfrac{(1+\sqrt{K})^2}{4},              & 1 < K < 2, \\[6pt]
				\dfrac{(K-1)(1+\sqrt{K})^2}{K^2},       & K \geq 2.
			\end{cases}
			\label{eq:highSNR_ratio}
		\end{align}
		Since $K > 1$ (i.e., $T_d > M$), the values of both piecewise functions above are strictly greater than $1$. This indicates that, under optimal power allocation, the high-SNR performance of the DAS scheme consistently outperforms that of the PS scheme.

		\subsubsection{Equal Power Allocation}
		
		Under the condition of $\rho_\tau = \rho_d = \rho$, it is straightforward from \eqref{eq:rho_eff_Trad_highSNR_general} to obtain the effective SNR for the PS scheme as:
		\begin{align}
			\rho_{\mathrm{eff}}^{\mathrm{PS}} \approx \frac{1}{2}\,\rho.
			\label{eq:rho_eff_Trad_highSNR_fixed}
		\end{align}
		For the DAS scheme, setting $\rho_\tau = \rho_d = \rho$ (in this case, $b \approx K\rho$) in \eqref{eq:D_RMT_highSNR}, its sensing distortion can be approximated as:
		\begin{align}
			D_{\mathrm{DAS}} 
			&\approx \frac{NM}{2\rho^2} \left[ \sqrt{K^2\rho^2+4\rho^2} - K\rho \right] \nonumber \\
			&= \frac{NM}{\rho} \left( \frac{\sqrt{K^2+4}-K}{2} \right).
		\end{align}
		Substituting this result into \eqref{eq:rho_eff_RMT_highSNR_general} yields
		\begin{align}
			\rho_{\mathrm{eff}}^{\mathrm{DAS}} \approx \frac{2}{2+\sqrt{K^2+4}-K}\,\rho.
			\label{eq:rho_eff_RMT_highSNR_fixed}
		\end{align}
		Since $\sqrt{K^2+4} - K < 2$ holds for any $K > 0$, we have
		\begin{align}
			\lim_{\rho\to\infty}
			\frac{\rho_{\mathrm{eff}}^{\mathrm{DAS}}}{\rho_{\mathrm{eff}}^{\mathrm{PS}}}
			= \frac{4}{2+\sqrt{K^2+4}-K} > 1,
		\end{align}
		This indicates that, under the equal power condition, the DAS scheme is likewise strictly superior to the PS scheme in the high-SNR regime.

		To facilitate an intuitive comparison, Table \ref{tab:highSNR_summary} summarizes the asymptotic effective SNRs for both the PS and DAS schemes under the two power strategies, along with their limiting behaviors as $K \to \infty$.
		\begin{table}[!t]
			\centering
			\caption{High-SNR Asymptotic Effective SNR }
			\label{tab:highSNR_summary}
			\renewcommand{\arraystretch}{1.8}
			\begin{tabular}{lcc}
				\toprule
				\textbf{Scheme \& Power} &
				$\boldsymbol{\rho_{\mathrm{eff}}}$ &
				\textbf{Limit} ($K \to \infty$) \\
				\midrule
				PS, Optimal Power &
				$\dfrac{1+K}{(1+\sqrt{K})^2}\,\rho$ &
				$\to \rho$ \\[8pt]
				DAS, Optimal Power &
				$\dfrac{K^2-1}{K^2}\,\rho$ &
				$\to \rho$ (fastest) \\[8pt]
				PS, Equal Power &
				$\dfrac{1}{2}\,\rho$ &
				Fixed loss of $1$\,bit \\[8pt]
				DAS, Equal Power &
				$\displaystyle \frac{2}{2+\sqrt{K^2+4}-K}\,\rho$  &
				$\to \rho$ \\
				\bottomrule
			\end{tabular}
		\end{table}
		
		From Table~\ref{tab:highSNR_summary}, the following conclusions can be drawn. For any finite $K$ and under the same power allocation strategy, the $\rho_{\mathrm{eff}}$ of the DAS scheme is strictly greater than that of the PS scheme. As $K \to \infty$, all configurations except the PS scheme with equal power allocation asymptotically approach $\rho$. Notably, the DAS scheme with optimal power allocation converges the fastest, at a rate of $\mathcal{O}(1/K^2)$, whereas the PS scheme with optimal power allocation converges at a slower rate of $\mathcal{O}(1/\sqrt{K})$.

		\section{Numerical Results}\label{sec:numerical_results}
		
		In this section, we validate the derived theoretical results through Monte Carlo simulations and evaluate the performance gains of the DAS scheme over the PS scheme.
		
		Unless otherwise stated, the default simulation parameters are as follows: number of antennas $M = N = 12$, coherence block length $T = 30$,  $\mathrm{SNR} = 5$\,dB. All curves are averaged over at least $3000$ independent channel realizations. The horizontal axis of the $R(D)$ curves represents the normalized sensing distortion $D/(NM)$, and the vertical axis represents the ergodic communication rate $R$ (in bit/s/Hz).
		
		\begin{figure}[!t]
			\centerline{\includegraphics[width=3.2in]{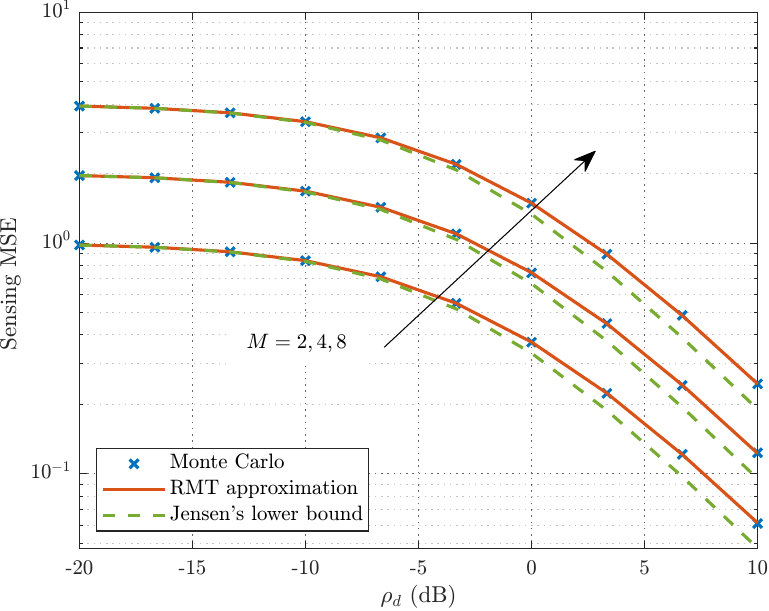}}
			\caption{ Sensing MSE versus $\rho_d$. Fixed parameters: $\gamma = 0.25$ and $c = 2$.}
			\label{fig:rmt_validation}
		\end{figure}
		Fig.~\ref{fig:rmt_validation} validates the accuracy of the RMT asymptotic expression in Proposition~\ref{prop:MSE_asymp}. As the data power $\rho_d$ increases, the sensing MSE monotonically decreases, indicating that the data symbols provide effective observation energy for the sensing task. The RMT approximation matches the Monte Carlo results exceptionally well across different numbers of antennas $M$. In contrast, although Jensen’s lower bound admits a concise form, it becomes noticeably loose in the high-SNR regime because it fails to capture the random spectral fluctuations of the data matrix.
		
			\begin{figure*}[t]
			\centering
			\begin{subfigure}[b]{0.48\textwidth}
				\centering
				\includegraphics[width=3.2in]{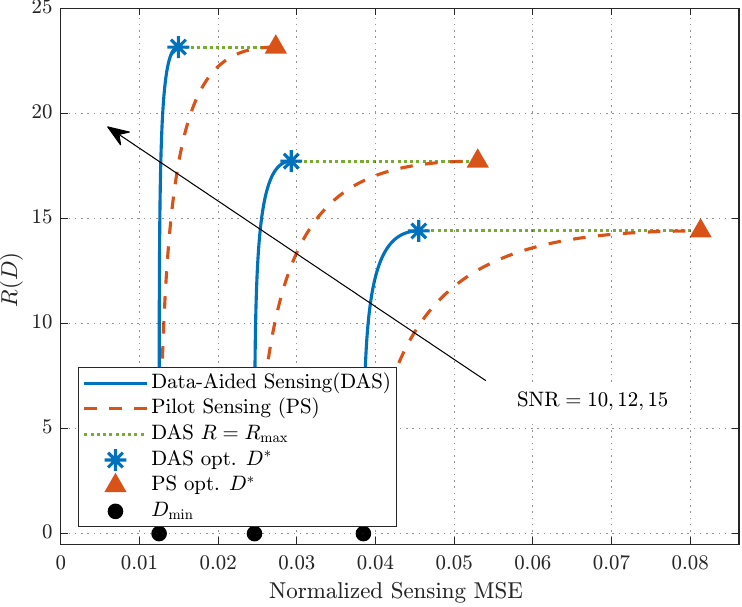}
				\caption{Optimal power allocation}
				\label{fig:SNR_impact_a}
			\end{subfigure}
			\hfill
			\begin{subfigure}[b]{0.48\textwidth}
				\centering
				\includegraphics[width=3.2in]{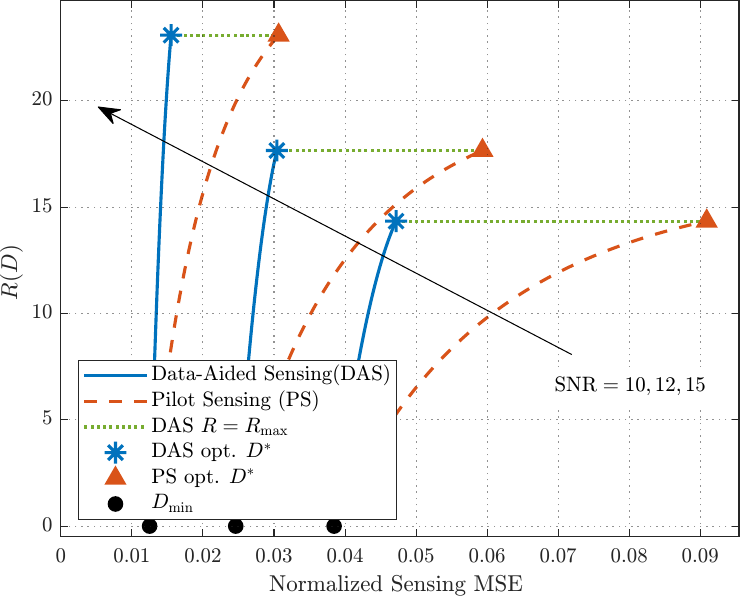}
				\caption{Equal power allocation}
				\label{fig:SNR_impact_b}
			\end{subfigure}
			\caption{$R(D)$ performance of the DAS and PS schemes under different transmit SNRs ($\{10, 12, 15\}$\,dB) with $T = 30$.}
			\label{fig:SNR_impact}
		\end{figure*}
		\begin{figure*}[t]
			\centering
			\begin{subfigure}[b]{0.48\textwidth}
				\centering
				\includegraphics[width=3.2in]{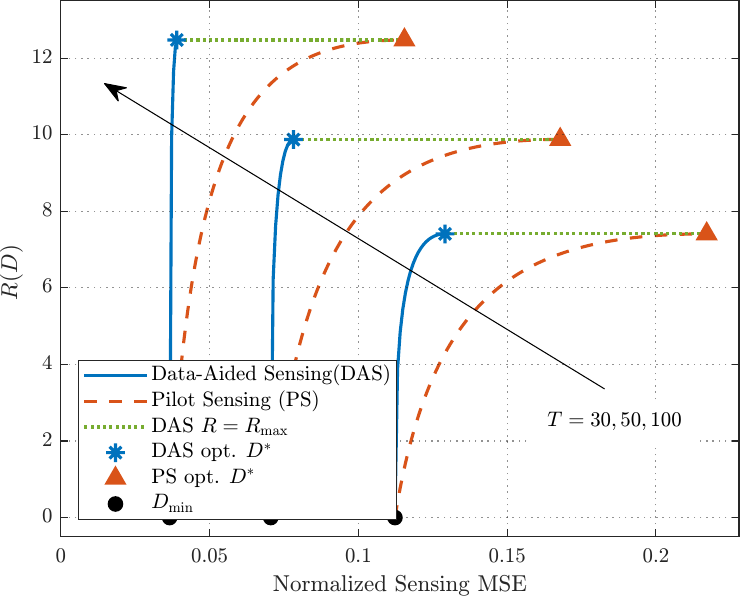}
				\caption{Optimal power allocation}
				\label{fig:T_impact_a}
			\end{subfigure}
			\hfill
			\begin{subfigure}[b]{0.48\textwidth}
				\centering
				\includegraphics[width=3.2in]{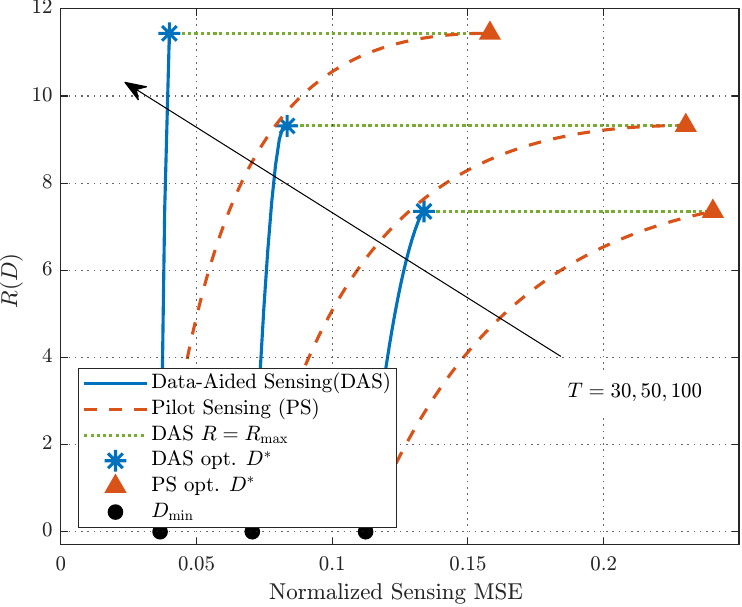}
				\caption{Equal power allocation}
				\label{fig:T_impact_b}
			\end{subfigure}
			\caption{$R(D)$ performance of the DAS and PS schemes under different coherence times ($T \in \{30, 50, 100\}$) with $\mathrm{SNR} = 5$\,dB.}
			\label{fig:T_impact}
		\end{figure*}
		
		Fig.~\ref{fig:SNR_impact} and Fig.~\ref{fig:T_impact} investigate the impact of the transmit SNR and the coherence time $T$ on the $R(D)$ performance of the DAS and PS schemes, respectively. Under any parameter configuration and power strategy, the DAS scheme consistently outperforms the PS scheme. This performance gain stems from the fact that DAS repurposes the communication data symbols as "pseudo-pilots," implicitly injecting additional observation energy for sensing. This effectively breaks the performance bottleneck of the PS scheme, which relies solely on dedicated pilots.

		As the SNR increases (Fig.~\ref{fig:SNR_impact}) or as $T$ grows (Fig.~\ref{fig:T_impact}), the $R(D)$ boundaries of the DAS scheme approach a rectangular shape. This indicates that the system can maintain the maximum communication rate $R_{\max}$ almost losslessly while simultaneously pushing the sensing distortion close to $D_{\min}$. In other words, the communication and sensing performances become asymptotically decoupled when resources are abundant.

		Fig.~\ref{fig:eq_vs_opt_power} investigates the impact of the coherence time $T$ on the power allocation strategies under the DAS scheme. When temporal resources are limited ($T=18$) or when temporal degrees of freedom are abundant ($T=30$), the optimal power allocation strategy can enhance the maximum communication rate by appropriately distributing power between pilots and data. Interestingly, when $T=2M$, the $R(D^*)$ points for both equal power allocation and optimal power allocation perfectly coincide.
		\begin{figure}[t]
			\centerline{\includegraphics[width=3.2in]{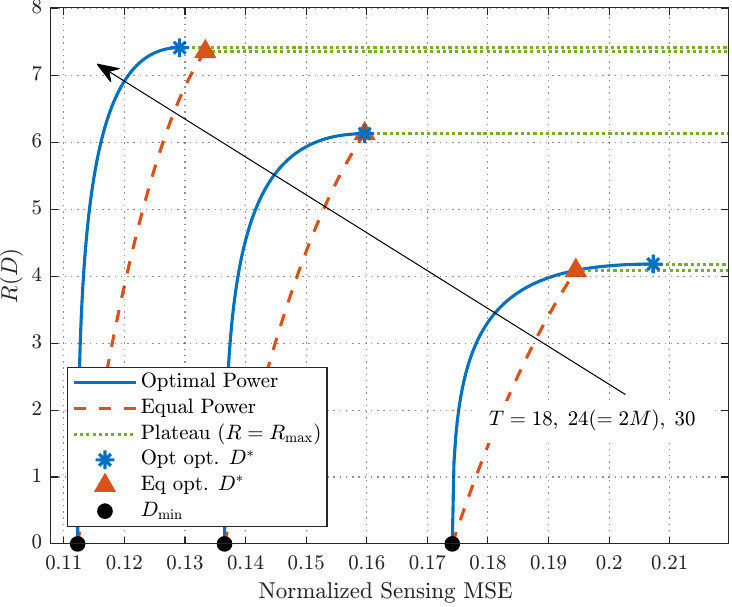}} 
			\caption{$R(D)$ performance of the DAS scheme under equal power allocation versus optimal power allocation, with $T \in \{18, 24, 30\}$.} 
			\label{fig:eq_vs_opt_power}
		\end{figure}
		
		Fig.~\ref{fig:snr_power_gain} presents the $R(D)$ comparison between the two power allocation strategies across different SNRs. As the SNR increases, the performance gain introduced by optimal power allocation gradually diminishes. This indicates that in the high-SNR regime, the bottleneck for the $R(D)$ performance shifts to the temporal resource constraint, rendering the marginal benefit of power optimization diminishing. Therefore, under high-SNR conditions, equal power allocation is sufficient to approach the performance of optimal power allocation.
		\begin{figure}[t]
			\centerline{\includegraphics[width=3.2in]{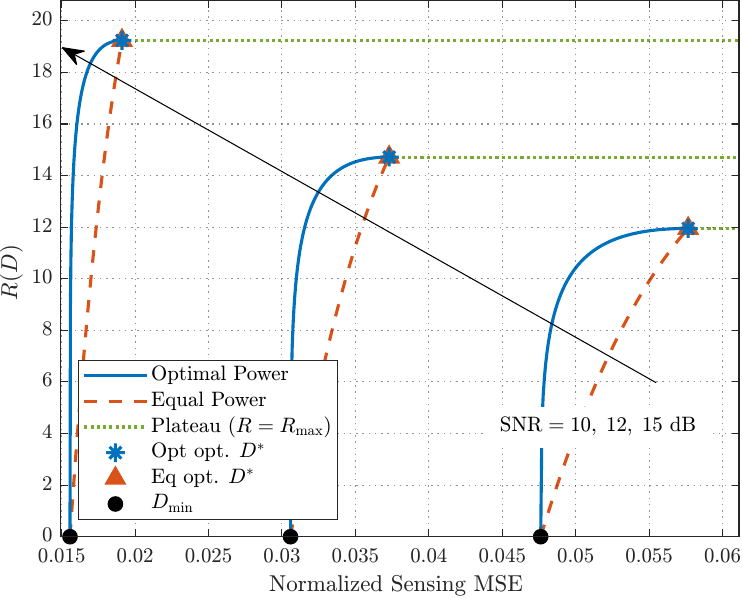}} 
			\caption{$R(D)$ comparison of the DAS scheme under equal power allocation versus optimal power allocation across different transmit SNRs ($\{10, 12, 15\}$\,dB) with $T = 24$.} 
			\label{fig:snr_power_gain}
		\end{figure}

		\section{Conclusion}\label{sec:conclusion}
		
		In this paper, we investigated the tradeoff between communication and sensing in a non-coherent P2P MIMO ISAC system, analyzing both the PS and DAS schemes. Based on RMT, a closed-form asymptotic expression for the sensing distortion of the DAS scheme was derived. Furthermore, the analytical forms of the $R(D)$ functions for both schemes were established under optimal power allocation and equal power allocation. The asymptotic analyses explicitly quantified the performance margin between the two schemes, revealing that the DAS scheme achieves a $3$ dB effective SNR improvement at low SNRs and an accelerated error decay rate at high SNRs.Moreover, in the limit of large coherence blocks or high SNRs, the communication rate and the sensing accuracy tend to be asymptotically decoupled.

		\appendices
		\section{Proof of proposition \ref{prop:optimal_pilot}}\label{app:A}
		Let $\mathbf{R}_{X_\tau} := \frac{1}{T_\tau}\mathbf{X}_\tau \mathbf{X}_\tau^H$. The optimization problem is then equivalent to
		\begin{subequations}\label{eq:pilot_opt_matrix}
			\begin{align}
				\min_{\mathbf{R}_{X_\tau} \succeq 0} \quad
				& N \operatorname{Tr}\left[
				\left(\mathbf{I}_M
				+ \frac{T_\tau}{\sigma_z^2}\mathbf{R}_{X_\tau}\right)^{-1}
				\right], \label{eq:pilot_opt_matrix_obj} \\
				\text{s.t.} \quad
				& \operatorname{Tr}(\mathbf{R}_{X_\tau}) \leq P_\tau.
				\label{eq:pilot_opt_matrix_cons}
			\end{align}
		\end{subequations}
		Applying the eigenvalue decomposition (EVD) to $\mathbf{R}_{X_\tau}$, let $\mathbf{R}_{X_\tau} = \mathbf{U}\boldsymbol{\Lambda}\mathbf{U}^H$, where $\mathbf{U}$ is a unitary matrix and $\boldsymbol{\Lambda} = \operatorname{diag}(\lambda_1, \dots, \lambda_M)$. Then, the objective function simplifies to $N\sum_{i=1}^M \frac{1}{1 + T_\tau\lambda_i/\sigma_z^2}$, and the optimization problem is transformed into
		\begin{subequations}\label{eq:pilot_opt_scalar}
			\begin{align}
				\min_{\{\lambda_i \geq 0\}} \quad
				& N\sum_{i=1}^M \frac{1}{1 + T_\tau\lambda_i/\sigma_z^2},
				\label{eq:pilot_opt_scalar_obj} \\
				\text{s.t.} \quad
				& \sum_{i=1}^M \lambda_i \leq P_\tau.
				\label{eq:pilot_opt_scalar_cons}
			\end{align}
		\end{subequations}
		
		Since the objective function is a convex function with respect to $\{\lambda_i\}$, according to Jensen's inequality, the optimal solution is achieved when $\lambda_1 = \cdots = \lambda_M = P_\tau/M$. That is, the optimal transmit covariance matrix is $\mathbf{R}_{X_\tau}^* = \frac{P_\tau}{M}\mathbf{I}_M$, which is equivalent to $\mathbf{X}_\tau\mathbf{X}_\tau^H = \frac{P_\tau T_\tau}{M}\mathbf{I}_M$. Substituting this into \eqref{eq:mse_expr} and utilizing $\rho_\tau = P_\tau/\sigma_z^2$, we directly obtain \eqref{eq:mse_opt}.

		\section{Proof of proposition \ref{prop:optimal_time}}\label{app:B}
		
		Under the sensing distortion constraint $D$, the total required power during the pilot phase is uniquely determined by $D$, which is given by
		\begin{equation}
			\rho_\tau T_\tau = M \frac{NM - D}{D}. \label{eq:app_pilot_energy}
		\end{equation}
		Substituting this into \eqref{eq:P1_energy}, the data power can be expressed as $\rho_d(T_d) = (\rho T - \rho_\tau T_\tau)/T_d$. Since $\rho_\tau T_\tau/M$ is a constant, the objective function simplifies to a single-variable function with respect to $T_d$:
		\begin{equation}
			R(T_d) = \frac{T_d}{T} h(\rho_d(T_d)), \label{eq:app_Ru_simplified}
		\end{equation}
		where $h(\rho_d) := \mathbb{E}\left[\log_2 \det\left(\mathbf{I}_N + \frac{\rho_{\mathrm{eff}}(\rho_d)}{M} \bar{\mathbf{H}}\bar{\mathbf{H}}^H\right)\right]$.
		
		Since $\rho_{\mathrm{eff}}(\rho_d)$ is monotonically increasing and strictly concave with respect to $\rho_d$, by the composition rules preserving concavity, $h(\rho_d)$ is also a strictly concave function and satisfies $h(0) = 0$. According to the tangent bound property of concave functions passing through the origin, for any $\rho_d > 0$, we have $h(\rho_d) > \rho_d h'(\rho_d)$. Taking the derivative of $R(T_d)$ with respect to $T_d$ yields:
		\begin{equation}
			\frac{\mathrm{d}R(T_d)}{\mathrm{d}T_d} = \frac{1}{T} \left[ h(\rho_d) - \rho_d h'(\rho_d) \right] > 0. \label{eq:app_derivative}
		\end{equation}
		Therefore, $R(T_d)$ is monotonically increasing with respect to $T_d$. Subject to the constraint $T_\tau \geq M$, the optimal solution is given by $T_\tau^* = M$.

		\section{Proof of Theorem \ref{thm:Rate_Distortion}}\label{app:C}
		Since the log-determinant function is strictly monotonically increasing and concave, the monotonicity and concavity of $R(D)$ are equivalent to the corresponding properties of $\rho_{\mathrm{eff}}(D)$. Expressing $\rho_{\mathrm{eff}}$ as a function of $\rho_\tau$:
		\begin{equation}
			f(\rho_\tau) = \frac{P\rho_\tau - M\rho_\tau^2}{A + B\rho_\tau},
		\end{equation}
		where $P = \rho T$, $A = T(1+\rho) - M > 0$, and $B = T - 2M$. Taking the second-order derivative of $f$ yields:
		\begin{equation}
			f''(\rho_\tau) = \frac{-2A(MA + BP)}{(A + B\rho_\tau)^3}.
		\end{equation}
		Noting that $MA + BP = (T-M)(M + \rho T) > 0$, and that $A + B\rho_\tau > 0$ always holds within the feasible region $\rho_\tau \in [0,\, \rho T/M]$, we have $f''(\rho_\tau) < 0$, meaning that $f$ is a strictly concave function. Setting $f'(\rho_\tau) = 0$ and taking the positive root yields the unique maximum point:
		\begin{equation}
			\rho_\tau^* = \frac{-A + \sqrt{A(A + BP/M)}}{B}.
		\end{equation}
		Substituting $\rho_\tau(D) = NM/D - 1$ into the expression above yields \eqref{eq:D_star}.
		
		For the composite function $\rho_{\mathrm{eff}}(D) = f(\rho_\tau(D))$, given that $\rho_\tau'(D) = -NM/D^2 < 0$ and $\rho_\tau''(D) = 2NM/D^3 > 0$, the chain rule gives:
		\begin{equation}
			\frac{\mathrm{d}^2 \rho_{\mathrm{eff}}}{\mathrm{d}D^2} = f''(\rho_\tau) \left(\rho_\tau'(D)\right)^2 + f'(\rho_\tau) \rho_\tau''(D). \label{eq:app_C_chain}
		\end{equation}
		When $D \in [D_{\min}^{(1)},\, D^{(1)*}]$, $\rho_\tau(D)$ is monotonically decreasing, which implies $\rho_\tau \geq \rho_\tau^*$, and thus $f'(\rho_\tau) \leq 0$. Combined with $f''(\rho_\tau) < 0$, both terms on the right-hand side of \eqref{eq:app_C_chain} are non-positive, resulting in $\frac{\mathrm{d}^2 \rho_{\mathrm{eff}}}{\mathrm{d} D^2} < 0$. Furthermore, since both $\rho_\tau(D)$ and $f$ are monotonically decreasing on $[\rho_\tau^*,\, \rho T/M]$, the composite function $\rho_{\mathrm{eff}}(D)$ is strictly monotonically increasing on $[D_{\min}^{(1)},\, D^{(1)*}]$. According to the composition properties preserving concavity, $R(D)$ is strictly concave and monotonically increasing over this interval. When $D > D^{(1)*}$, we have $f'(\rho_\tau) > 0$, and $\rho_{\mathrm{eff}}(D)$ becomes monotonically decreasing. Therefore, $D^{(1)*}$ is the optimal tradeoff point.

		\section{Proof of Theorem \ref{thm:Rate_Distortion_equal_power}}\label{app:D}
		
		Let $R(T_\tau) = \frac{T - T_\tau}{T}h(T_\tau)$. Taking the derivatives with respect to $T_\tau$ yields
		\begin{align}
			R'(T_\tau) &= -\frac{1}{T} h(T_\tau) + \frac{T - T_\tau}{T} h'(T_\tau), \\
			R''(T_\tau) &= -\frac{2}{T} h'(T_\tau) + \frac{T - T_\tau}{T} h''(T_\tau).
		\end{align}
		Since $\rho_{\mathrm{eff}}(T_\tau)$ is strictly  increasing and concave with respect to $T_\tau$, and the log-determinant function preserves concavity, it follows that $h(T_\tau)$ is also  strictly increasing and concave. That is, $h'(T_\tau) > 0$ and $h''(T_\tau) < 0$. Under the physical constraint $T_\tau \leq T$, we have $R''(T_\tau) < 0$ globally. Therefore, $R(T_\tau)$ is a strictly concave function, and there exists a unique maximum point $T_\tau^*$ satisfying $R'(T_\tau^*) = 0$.
		
		For the composite function $R(D) = R(T_\tau(D))$, from \eqref{eq:T_tau_D}, we have $T_\tau'(D) = -NM^2/(\rho D^2) < 0$ and $T_\tau''(D) = 2NM^2/(\rho D^3) > 0$. The chain rule gives
		\begin{equation}
			\frac{\mathrm{d}^2 R(D)}{\mathrm{d}D^2}
			= R''(T_\tau)\!\left(T_\tau'(D)\right)^2
			+ R'(T_\tau)\,T_\tau''(D).
			\label{eq:app_D_chain}
		\end{equation}
		When $D \in [D_{\min}^{(2)}, D^{(2)*}]$, since $T_\tau(D)$ is monotonically decreasing, we have $T_\tau \geq T_\tau^*$, which implies $R'(T_\tau) \leq 0$. Combined with the facts that $R''(T_\tau) < 0$ and $T_\tau''(D) > 0$, both terms on the right-hand side of \eqref{eq:app_D_chain} are non-positive. Consequently, $\frac{\mathrm{d}^2 R(D)}{\mathrm{d}D^2} < 0$ , meaning that $R(D)$ is strictly concave on the interval $[D_{\min}^{(2)}, D^{(2)*}]$.

		\section{Proof of proposition \ref{prop:lower_bound}}\label{app:E}
		
		Assuming i.i.d. Gaussian signaling with equal power allocated across antennas, the covariance matrix satisfies $\mathbb{E}[\mathbf{X}_d \mathbf{X}_d^H] = \frac{P_d T_d}{M} \mathbf{I}_M$. Let $\mathbf{R} := \frac{1}{\sigma_z^2} \mathbf{X}_d \mathbf{X}_d^H$; it then follows that $\mathbb{E}[\mathbf{R}] = \frac{\rho_d T_d}{M} \mathbf{I}_M$. The sensing MSE can be written as
		\begin{equation}
			\mathrm{MSE} = N \cdot \mathbb{E}_{\mathbf{R}} \left[ \mathrm{Tr} \left( (c\,\mathbf{I}_M + \mathbf{R})^{-1} \right) \right]. \label{eq:app_MSE_R}
		\end{equation}
		Since the function $f(\mathbf{A}) = \operatorname{Tr}(\mathbf{A}^{-1})$ is strictly convex over the cone of positive definite matrices, applying  Jensen's inequality to \eqref{eq:app_MSE_R} yields
		\begin{align}
			\mathrm{MSE} 
			&\ge N \cdot \mathrm{Tr}\left( \left(c\,\mathbf{I}_M + \mathbb{E}[\mathbf{R}]\right)^{-1} \right) \notag \\
			&= N \cdot \mathrm{Tr}\left( \left( \left(c + \frac{\rho_d T_d}{M}\right) \mathbf{I}_M \right)^{-1} \right) \notag \\
			&= \frac{NM^2}{Mc + \rho_d T_d}. \label{eq:app_lb}
		\end{align}
		From the definition of $c$, we have $Mc = M + \rho_\tau T_\tau$. Utilizing the power constraint $\rho_\tau T_\tau + \rho_d T_d = \rho T$, the denominator in \eqref{eq:app_lb} simplifies to $M + \rho T$, which directly leads to \eqref{eq:MSE_lower_bound}.

		\section{Proof of proposition \ref{prop:DA_optimal_time}}\label{app:F}
		
		Assume, to the contrary, that the optimal policy is $\Theta' = \{T_\tau', \rho_\tau', \rho_d'\}$ with $T_\tau' > M$. We can construct a new policy $\tilde{\Theta} = \{M, \tilde{\rho}_\tau, \tilde{\rho}_d\}$ that preserves the total power expended in both phases:
		\begin{align}
			E_\tau := \tilde{\rho}_\tau M = \rho_\tau' T_\tau', \qquad
			E_d    := \tilde{\rho}_d(T-M) = \rho_d'(T-T_\tau').
		\end{align}
		
		\textit{Sensing constraint.} Under fixed parameters $c$ and $E_d/M$, $D_{\mathrm{RMT}}$ is strictly monotonically increasing with respect to $\rho_d$. Since the data power of the new policy satisfies $\tilde{\rho}_d = E_d/(T-M) < \rho_d'$, it follows that $\tilde{D}_{\mathrm{RMT}} < D_{\mathrm{RMT}}' \leq D$. Consequently, the sensing distortion constraint is strictly satisfied.
		
		\textit{Communication performance.} Under a fixed $E_d$, the achievable rate $R(T_d) = \frac{T_d}{T} h\left(\frac{E_d}{T_d}\right)$ is strictly monotonically increasing with respect to $T_d$ (as shown in Appendix B). Given that $\tilde{T}_d = T - M > T - T_\tau' = T_d'$, we have $\tilde{R} > R'$.
		
		Therefore, the constructed policy $\tilde{\Theta}$ strictly outperforms $\Theta'$ while satisfying all the constraints. This contradicts the presumed optimality of $\Theta'$, thereby proving that $T_\tau^* = M$.

		\section{Proof of Theorem \ref{thm:Rate_Distortion_DA}}\label{app:G}
		
		Substituting \eqref{eq:rho_tau_rho_d} into \eqref{eq:rho_eff_reduced}, $\rho_{\mathrm{eff}}$ can be simplified into a rational function with respect to $\rho_d$:
		\begin{equation}
			g(\rho_d) = \frac{P\rho_d - C\rho_d^2}{A + B\rho_d},
		\end{equation}
		where $P = \rho T$, $A = M + \rho T > 0$, $B = 2M-T$, and $C = T-M > 0$. Since this function shares the exact same rational structure as $f(\rho_\tau)$ in Appendix \ref{app:C}, following the concavity arguments presented in Appendix \ref{app:C}, we obtain $g''(\rho_d) < 0$. The unique maximum point is thus given by
		\begin{equation}
			\rho_d^* = \frac{-(M+P)
				+ \sqrt{(M+P)^2 + \frac{(M+P)(2M-T)P}{T-M}}}{2M-T}.
		\end{equation}
		Substituting $\rho_d^*$ into \eqref{eq:D_RMT_reduced} yields $D^{(3)*}= D_{\mathrm{RMT}}(\rho_d^*)$.
		
		Next, we derive the inverse function $\rho_d(D)$. Let $\kappa := NM/D$ and $u := 1 + \rho T/M$. By rearranging the equation $D_{\mathrm{RMT}}(\rho_d) = D$, we obtain a quadratic equation with respect to $\rho_d$:
		\begin{equation}
			\frac{T-M}{M}\rho_d^2 - (u - \kappa)\rho_d - \kappa(u - \kappa) = 0.
		\end{equation}
		Taking the positive root directly yields \eqref{eq:rho_d_D}. Based on the monotonicity of $D_{\mathrm{RMT}}(\rho_d)$, it is evident that $\rho_d(D)$ is also strictly monotonically increasing with respect to $D$.
		
		Finally, we analyze the concavity of $\rho_{\mathrm{eff}}(D) = g(\rho_d(D))$. Applying the chain rule, we have
		\begin{equation}
			\frac{\mathrm{d}^2\rho_{\mathrm{eff}}}{\mathrm{d}D^2}
			= g''(\rho_d)\bigl(\rho_d'(D)\bigr)^2 + g'(\rho_d)\,\rho_d''(D).
		\end{equation}
		When $D \in [D_{\min}^{(3)} , D^{(3)*}]$, we have $\rho_d \le \rho_d^*$, which implies $g'(\rho_d) \ge 0$. Combined with $g''(\rho_d) < 0$ and the fact that $\rho_d''(D) < 0$ (obtained by taking the second-order derivative of \eqref{eq:rho_d_D}), both terms on the right-hand side of the above equation are non-positive. Therefore, $\rho_{\mathrm{eff}}(D)$ is strictly concave and monotonically increasing over this interval. By the composition properties preserving concavity, $R(D)$ exhibits the same properties. When $D > D^{(3)*}$, $g'(\rho_d) < 0$ and $\rho_{\mathrm{eff}}(D)$ becomes monotonically decreasing, establishing $D^{(3)*}$ as the optimal tradeoff point.

	\end{document}